\documentclass[reprint,nofootinbib,amsmath,amssymb,aps,pra,onecolumn
]{revtex4-2}

\makeatletter
\AtBeginDocument{\let\LS@rot\@undefined}
\makeatother

\usepackage{color}
\usepackage{physics}
\usepackage{soul}
\usepackage[print-unity-mantissa=false]{siunitx}
\usepackage{graphicx}
\usepackage{bm}
\usepackage{amsthm}
\usepackage{thmtools}
\usepackage{pdfpages} 
\usepackage{pgffor}  

\usepackage{tabularx}
\usepackage{hyperref}
\hypersetup{
    colorlinks=true,
    linkcolor=blue
    }

\newcounter{protocol}
\newenvironment{protocol}[1]
    {\par\addvspace{\topsep}
   \noindent
   \tabularx{\linewidth}{@{} X @{}}
    \hline \textbf{Protocol} \refstepcounter{protocol}\textbf{\theprotocol} #1 \\
    \hline}
  { \\
    \hline
   \endtabularx
   \par\addvspace{\topsep}
    }

\newtheorem{theorem}{Theorem}
\newtheorem{lemma}[theorem]{Lemma}
\newtheorem{corollary}[theorem]{Corollary}

\newtheorem{definition}{Definition}

\newenvironment{customthm}[1]{%
  \par\noindent\textbf{Corollary #1.}\itshape\ }
  {\par\normalfont}

\newcommand{\vars}[1]{\mathcal{#1}}
\newcommand{\channel}{\mathcal{E}}

\newcommand{\Hil}{\mathcal{H}}

\newcommand{\cnorm}{\hat{C}}
\newcommand{\etath}{\eta_{\mathrm{th}}}
\newcommand{\etatol}{\eta_{\mathrm{tol}}}

\newcommand{\ctolnorm}{\hat{C}_{\mathrm{tol}}}
\newcommand{\cth}{C_{\mathrm{th}}}
\newcommand{\cthnorm}{\hat{C}_{\mathrm{th}}}

\begin{document}

\title{Efficient Multi-basis Quantum Position Verification Secure against Generalized Adversaries}






\author{Wen Yu Kon}
\affiliation{Global Technology Applied Research, JPMorganChase}

\author{Ignatius William Primaatmaja}
\affiliation{Global Technology Applied Research, JPMorganChase}

\author{Kaushik Chakraborty}
\affiliation{Global Technology Applied Research, JPMorganChase}

\author{Charles Lim}
\affiliation{Global Technology Applied Research, JPMorganChase}

\date{\today}

\begin{abstract}
Quantum position verification (QPV) enables multiple verifiers to certify a prover’s location using quantum communication and physical assumptions. With experimental demonstrations of QPV becoming increasingly feasible, enhancing the practicality and security of QPV protocols is more important than ever. In this work, we make three key contributions toward this goal. First, we introduce a robust QPV protocol in which the verifier’s state preparation is independent of channel loss, improving reliability in real-world conditions. Second, we refine existing security analysis techniques to bolster protocol resilience against experimental imperfections. Third, we identify and address some implicit assumptions present in existing security analyses, providing a framework to eliminate such assumptions. Additionally, as an example of QPV application beyond location verification, we illustrate how QPV can be leveraged for authenticating classical communication in quantum key distribution.
\end{abstract}

\maketitle

\section{Introduction}

Certifying the location of a party can be important for many applications such as location-based services, fraud prevention, and geographic access control.
While classical position verification protocols exist, they require assumptions such as pre-shared key or bounded classical memory to achieve security~\cite{Chandran2009}.
Therefore, \emph{quantum position verification} (QPV) has attracted significant attention from the community due to its reliance on weaker physical assumptions.
Following the first proof of security against linearly entangled adversaries~\cite{Bluhm2022}, subsequent work has established security under various conditions, including slow quantum communication~\cite{Llorenc2023,Llorence2025_Lossy}, channel loss tolerance~\cite{ABB23}, and adversaries with linear quantum gates~\cite{Asadi2025}.
More recently, preliminary results for the implementation of SWAP QPV have been presented~\cite{Kanneworff2025}, signaling interest to bring the protocol from theory to reality. Consequently, enhancing the practicality and security of QPV protocols has become increasingly critical.

One important challenge for QPV is its sensitivity to loss, with QPV protocols relying on BB84 states tolerating at most \SI{50}{\percent} loss~\cite{Bluhm2022,Llorenc2023}.
To improve loss tolerance, QPV protocols with matching state preparation and measurement basis~\cite{Llorenc2023} have been proposed, where the loss tolerance scales as $\sim 1/m$ with $m$ being the number of basis choices.
In this work, we propose a multi-basis QPV protocol where six states ($\{\ket{0},\ket{1},\ket{+},\ket{-},\ket{+i},\ket{-i}\}$) are prepared, but with multiple measurement bases.
This design offers three advantages: (1) reduced experimental complexity due to fewer prepared states, (2) enhanced flexibility since state preparation and measurement bases need not match, and (3) elimination of the requirement for the verifiers to securely communicate the basis choice for state preparation.
Importantly, these advantages come without additional penalty, since the proposed protocol matches the performance of earlier proposals in Ref.~\cite{Llorenc2023}.

To prove security of the proposed protocol against entangled adversaries, we adapt and refine the analysis presented in Ref.~\cite{Llorenc2023}.
We introduce two modifications that can ease the implementation requirements of QPV.
Firstly, we tighten the lower bound on the trace distance between two sets of states by formulating it as a semidefinite program (SDP) using methods from Ref.~\cite{wang2019characterising} and linear approximation.
Secondly, we provide a more precise construction of classical rounding with a fixed size of the output set for multi-basis QPV.
Together, these improvements can increase the error tolerance, thereby easing QPV implementation.

While the security analysis adapted from Ref.~\cite{Llorenc2023} covers a broad class of adversaries, it relies on several implicit assumptions.
In this work, we generalize the adversary model by addressing three assumptions: (1) adversaries use only pure states and unitary ($q$-qubit pure strategy), (2) adversaries do not share unbounded classical randomness, and (3)  transmission $\eta$ is independent of inputs $(x,y)$ and shared randomness.
The first assumption is addressed by purifying the attack strategy, which allows us to map any mixed state strategy to an equivalent pure state strategy.
The second and third assumptions are addressed through a different partitioning strategy.
Unlike the original analysis which partitions events based solely on error rate, we add further partitions based on transmission, yielding 4 partitions having high/low error with high/low transmission. While this allows for proper accounting of transmission and shared randomness, we find that loss tolerance degrades significantly compared to the original adversary model.

QPV can have applications beyond location verification. 
Here, we illustrate one such application: using QPV as a means of authentication in \emph{quantum key distribution} (QKD)~\cite{Buhrman2014}, either when standard Wegman-Carter authentication~\cite{Carter1977_WegmanCarter} fails or when location-based credentials are required.

Recently, Ref.~\cite{Llorenc2025} introduced a tightened security
analysis for QPV without qubit loss and achieved a
higher error tolerance. Extending their approach to lossy channels and incorporating our generalizations and improvements remain open questions.

For clarity, we summarize our main contributions in this manuscript relative to Ref.~\cite{Llorenc2023,Llorenc2025}. \textbf{Protocol-level novelty:} We introduce a multi-basis QPV protocol which uses only six prepared states and multiple measurement bases that utilize the asymmetric Bell expectation score. This reduces experimental complexity while enhancing flexibility by decoupling state preparation from measurement basis without any compromise in security. \textbf{Proof-level novelty:} We develop a refined security analysis by tightening trace-distance bounds, and a modified classical rounding argument, leading to an improved error and loss tolerance. \textbf{Model-level novelty:} We generalize the adversary model to allow for mixed states, unbounded shared randomness, and input-dependent transmission, exposing limitations of prior analyses and clarifying the scope of our rigorous security guarantees.

\section{Results}

\subsection{Quantum Position Verification}

\begin{figure}
    \centering
    \includegraphics[width=0.75\linewidth]{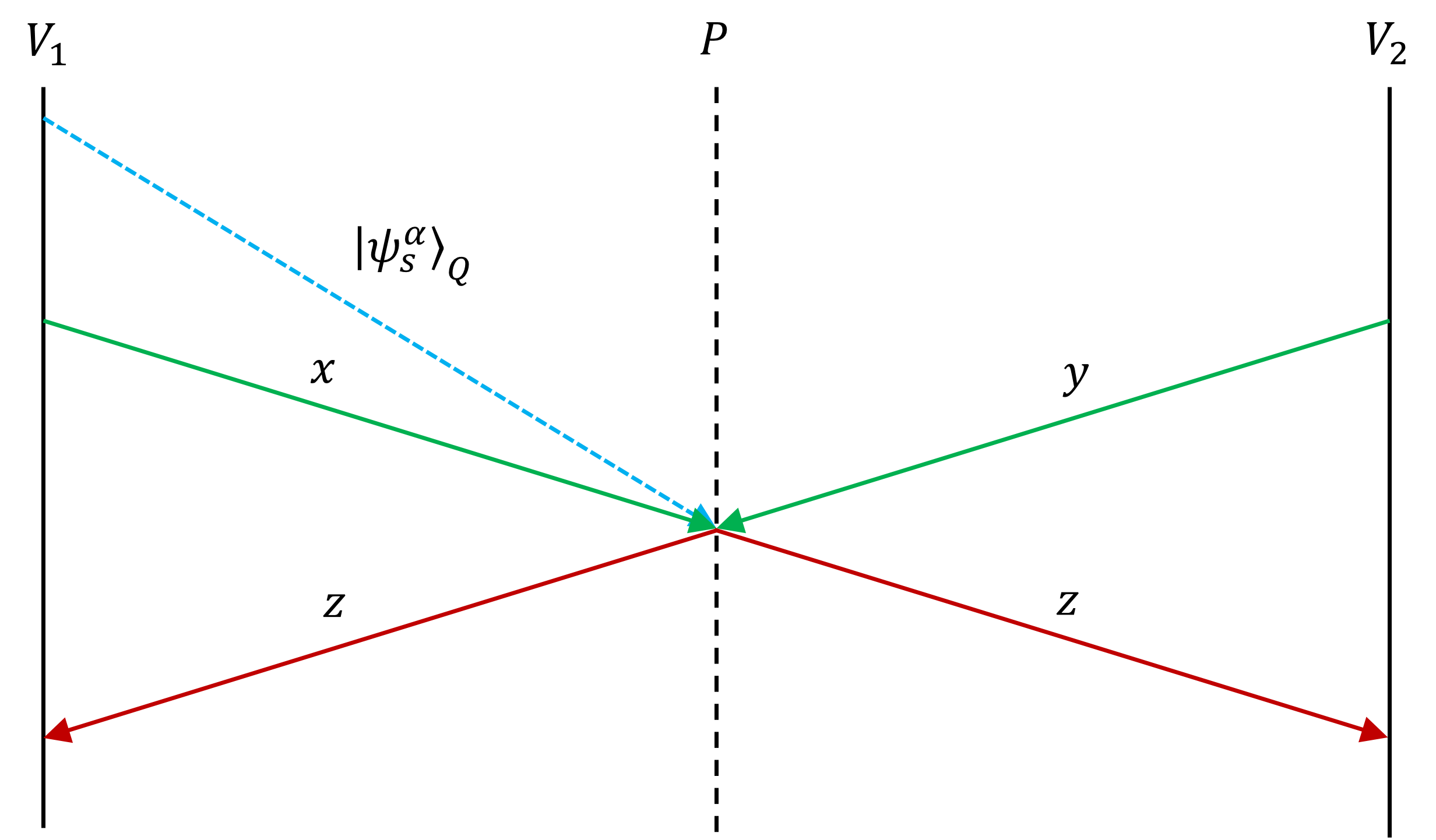}
    \caption{Spacetime diagram illustrating the classical and quantum communication in $QPV_{BB84}^{\eta,f}$ QPV protocol. Note that the solid lines represent classical communication which occurs at the speed of light while the dotted lines represent quantum communication.}
    \label{fig:QPV_Original_Protocol}
\end{figure}

The goal of QPV is to certify that a prover is at a specified location $P$. For simplicity, we consider the scenario where $P$ is co-linear with two trusted verifiers $V_1$ and $V_2$ who share a secure communication channel.
A loss-tolerant QPV protocol with transmission $\eta$ can be used to certify the prover's location, labeled $QPV_{BB84}^{\eta,f}$~\cite{Llorenc2023} (illustrated in Figure~1):
\renewcommand{\labelenumi}{Step \arabic{enumi}.}
\begin{enumerate}
    \item Verifier Preparation: Verifiers $V_1$ and $V_2$ randomly select $x\in\{0,1\}^n$ and $y\in\{0,1\}^n$. $V_2$ sends $y$ to $V_1$ via a secure communication channel. Verifiers also agree on a target arrival time $t_P$ for their classical messages.
    \item \textbf{Quantum State Transfer}: The first verifier computes $\alpha=f(x,y)$ using evaluation function $f$ and randomly selects $s\in\{0,1\}$. The first verifier then prepares a BB84 state $\ket{\psi^{\alpha}_{s}}_Q$ with basis $\alpha$ and bit value ${s}$ and sends it to the prover.
    \item \textbf{Verifier Message}: Verifiers 1 and 2 send $x$ and $y$ {respectively} such that they both arrive at $P$ at time $t_P$.
    \item Prover Measurement: The prover computes $\alpha=f(x,y)$ and measures the quantum system $Q$ in basis $\alpha$. The prover immediately sends the outcome $z$ ($z=\perp$ if no photon is detected) to both verifiers.
    \item \textbf{Timing and Validity Check}: The verifiers check if the time they receive the response $z$ is within a certain time threshold and if the responses between verifiers match. 
    \item \textbf{Parameter Estimation}: Repeat steps 1 to 5 for multiple rounds, and compute the error rate (the proportion of rounds where $s\neq z$) and the transmission  $\eta$ (the proportion of the rounds where $z \neq \perp$). The prover's position is verified if the error rate is below some threshold and all timing and validity checks pass.
\end{enumerate}
\renewcommand{\labelenumi}{\arabic{enumi}.}
The security of the QPV protocol relies on the observation that (1) parties not at location $P$ cannot obtain $x$ and $y$ before the party at location $P$ and (2) quantum information in $Q$ cannot be duplicated, thereby preventing parties not at $P$ from providing correct responses $z$ to both verifiers.

\subsection{Modeling the adversaries' strategy}
For QPV protocols such as the $QPV^{\eta, f}_{BB84}$ protocol, we can without any loss of generality (WLOG) consider two adversaries, Alice and Bob, located to the left and to the right of position $P$, respectively. The most general attack that Alice and Bob can perform proceeds as follows~\cite{Llorenc2023} (summarized in Figure~2):
\renewcommand{\labelenumi}{Step \arabic{enumi}.}
\begin{enumerate}
    \item Pre-share initial state: Alice and Bob pre-share the quantum state $\rho_{AB}$ where the register $A$ is kept in Alice's location and the register $B$ is kept in Bob's location. In general, the state $\rho_{AB}$ may be entangled.
    \item Local operations: After verifier $V_1$ sends the quantum register $Q$ and the input $x$, Alice applies a \textit{completely positive and trace preserving} (CPTP) map $\vars{E}_{QA}^x$ on the registers $Q$ and $A$. Similarly, Bob applies the CPTP map $\vars{E}_B^{y}$ on the register $B$ after the verifier $V_2$ sends the input $y$. Note that the choice of CPTP map can depend on the input accessible to each party.
    \item Communication: Alice splits her register into two separate sub-systems $A_k$ and $A_s$ and Bob similarly splits his register into $B_k$ and $B_s$. Alice and Bob keep the registers $A_k$ and $B_k$ at their respective locations, while exchanging sub-systems $A_s$ and $B_s$ along with a copy of the inputs $x$ and $y$. (Subscript $k$: keep, $s$: send)
    \item Measurement: After receiving both inputs $x$ and $y$ and the exchanged quantum sub-systems, Alice and Bob perform local measurements on their respective quantum registers. These measurements can be described by \textit{positive operator-valued measures} (POVMs) $\{A^{xy}_z\}_z$ and $\{B^{xy}_z\}_z$.
\end{enumerate}
\begin{figure}
    \centering
    \includegraphics[width=0.75\linewidth]{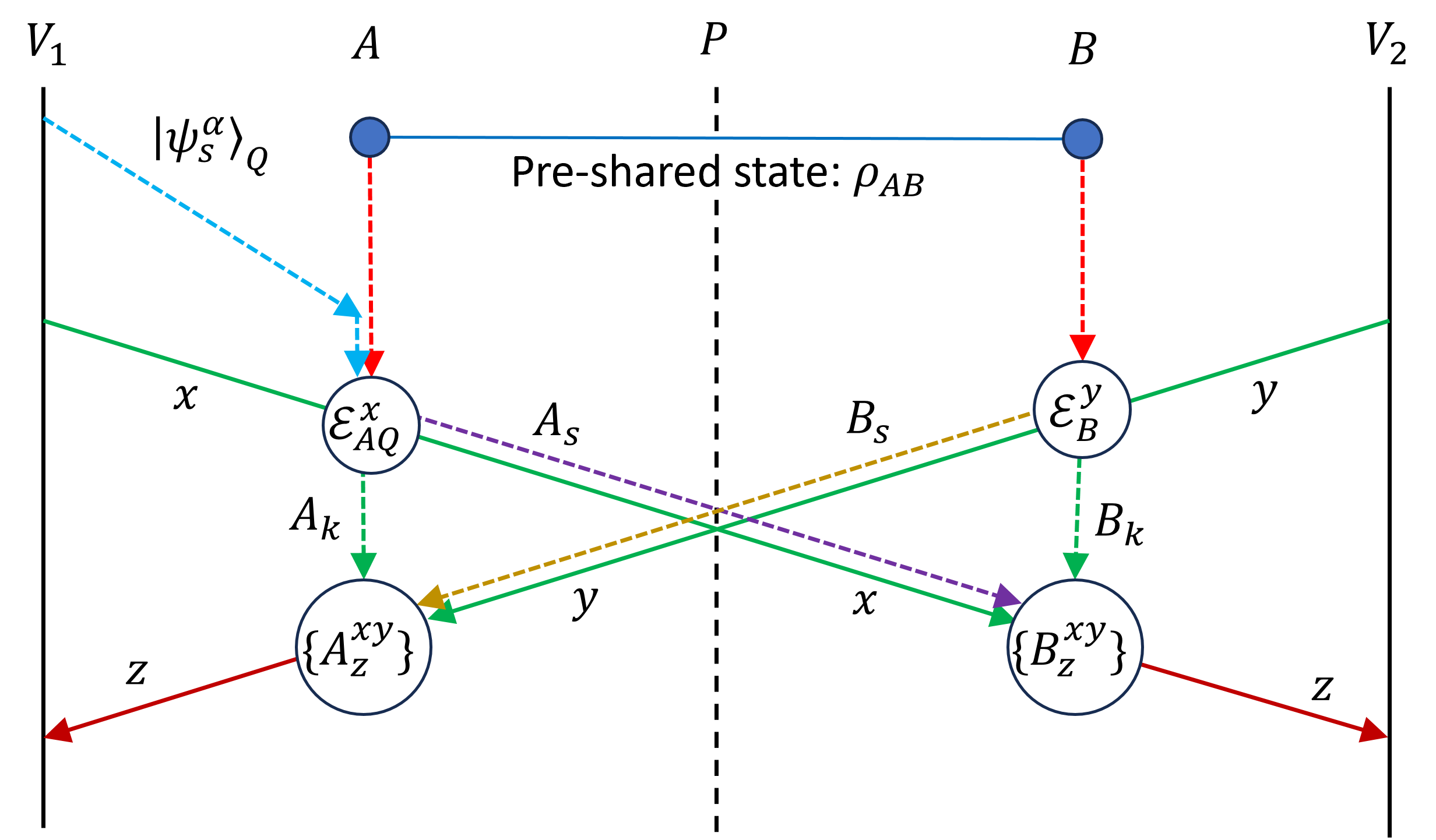}
    \caption{Spacetime diagram illustrating the generalized attack strategy. The dotted lines represent quantum information while solid lines represent classical information transmitted. The circles $\channel_{AQ}^x$ and $\channel_B^y$ represent the Alice's and Bob's location operations while the circles $\{A^{xy}_z\}$ and $\{B^{xy}_z\}$ represent their measurements. Note that lines $A_s$ and $x$, and $B_s$ and $y$ are overlapping (they are presented as split for clarity).}
    \label{fig:Main_Protocol_Attack}
\end{figure}
\renewcommand{\labelenumi}{\arabic{enumi}.}
The adversary strategy can thus be characterized using the tuple $$\vars{A} = \{\rho_{AB}, \{\vars{E}_{QA}^x\}_x, \{\vars{E}_B^{y}\}_y, \{A_{a}^{xy}\}_{a, x,y},\{B_b^{xy}\}_{b,x,y}, \vars{S}_{A\rightarrow A_kA_s}, \vars{S}_{B\rightarrow B_kB_s}\},$$
where $\vars{S}_{A\rightarrow A_kA_s}$ and $\vars{S}_{B\rightarrow B_kB_s}$ are quantum channels describing the splitting of Alice's and Bob's registers.
It is well known that when endowed with an unbounded amount of entangled states, the adversaries can perform an attack based on \textit{port-based teleportation} to break any QPV protocol~\cite{Beigi2011,Buhrman2014}. Therefore, QPV protocols can only be secure when considering restricted adversaries. In this work, we restrict ourselves to the case where Alice and Bob have bounded quantum memory, which in turn gives an upper bound on the amount of entanglement pre-shared by Alice and Bob.

When analyzing security against entangled adversaries, we consider an adversarial model similar to that of Ref.~\cite{Llorenc2023}, except that we adjust it to account for parameter estimation based on scores (instead of error rates), as well as tightening the security bounds. We call this adversarial strategy a $q$\textit{-qubit restricted strategy} as this model considers adversaries which are limited to pure $q$-qubit pre-shared quantum states and they can only perform local unitary operations. In this model, the adversaries have the additional constraint that the ``no-detection'' response is independent of the inputs $x$ and $y$. In contrast, in a later section generalizing the security against entangled adversaries, we allow the adversaries to pre-share mixed quantum states (with a bounded dimension) as well as an unbounded amount of pre-shared randomness. We also allow the adversaries to apply any local CPTP maps that may depend on the shared randomness. We also remove the additional assumption on the ``no-detection'' response. We call such a strategy a $q$-\textit{qubit general strategy}. Table 1 below highlights the differences between these adversarial models.

\begin{table*}[h]
    \centering
    \includegraphics[width=\linewidth]{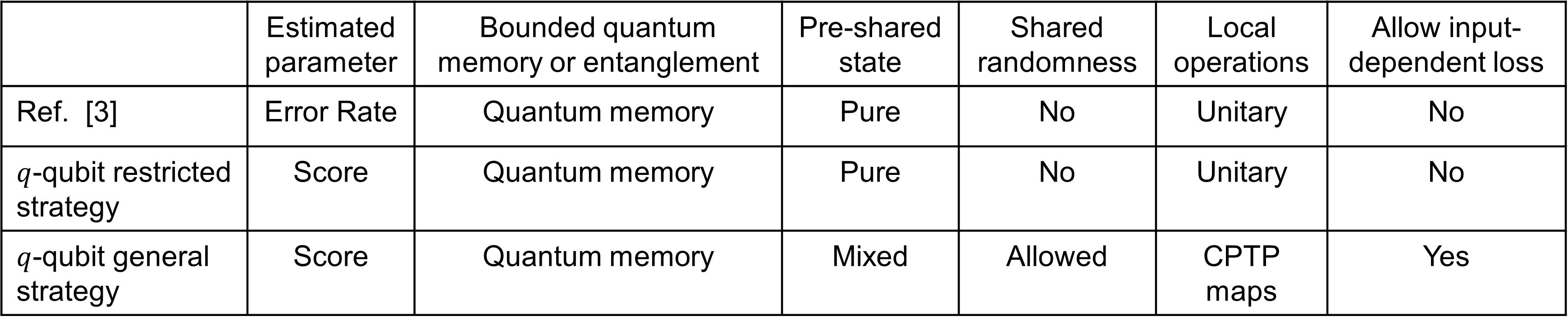}
    \caption{Comparison between the adversarial models in Ref.~\cite{Llorenc2023}, the $q$-qubit pure strategy (in the section on security against entangled adversaries) and the $q$-qubit mixed strategy (in the section on generalization of security against entangled adversaries).}
    \label{tab: adversarial model comparison}
\end{table*}

\subsection{Simple Multi-basis QPV}

In general, QPV protocols where a single slow quantum state is sent to the prover have loss tolerance that scales as $\eta\sim 1/m$, where $m$ is the number of basis choices.
Using the BB84 states and measurements leads to a \SI{50}{\percent} loss tolerance, which may be challenging to implement in practice.
Consequently, the use of multi-basis QPV protocols, such as that proposed in Ref.~\cite{Llorenc2023} with matching preparation and measurement basis distributed over the Bloch sphere (along with a change in $f$ to output values in $[m]=\{1,\cdots,m\}$), may be necessary.
The $\eta \sim 1/m$ loss tolerance is due to a general attack where the adversary can randomly select a basis to measure, and post-selects on rounds where the measured basis matches the computed basis. Interestingly, this attack only scales with the number of measurement bases, but not the number of prepared states.

We therefore propose a multi-basis QPV (Protocol 1) where six states ($\{\ket{0},\ket{1},\ket{+},\ket{-},\ket{+i},\ket{-i}\}$) are prepared and sent, and $m$ measurement bases on the Bloch sphere are utilized.
One could also consider a variant in which the verifier prepares the BB84 states and the prover's measurement bases lie in the XZ-plane of the Bloch sphere. However, this variant generally has lower noise tolerance, despite being slightly simpler to implement.

Unlike the multi-basis QPV protocol presented in Ref.~\cite{Llorenc2023}, our protocol does not require the verifier's preparation basis to match the prover's measurement basis. Since the basis choices can be mismatched, the prover's measurement outcome need not perfectly correlate with the verifier's bit value, even in the absence of noise. Therefore, error rate is no longer a suitable measure of the correlation between the verifier and the prover, and we have to generalize to the notion of a score.

Let $\alpha$ be the verifier's preparation basis and $\alpha'$ be the prover's measurement basis with basis vectors
\begin{equation}
    \begin{split}
        \ket{\psi^{\alpha'}_0} &=\cos \left(\frac{\theta_{\alpha'}}{2} \right)\ket{0}+\sin \left(\frac{\theta_{\alpha'}}{2} \right)e^{i\varphi_{\alpha'}}\ket{1},\\
        \ket{\psi^{\alpha'}_1} &= \sin\left(\frac{\theta_{\alpha'}}{2}\right) \ket{0} - \cos \left(\frac{\theta_{\alpha'}}{2}\right) e^{i \varphi_{\alpha'}} \ket{1},
    \end{split}
\end{equation}
where $\theta_{\alpha'}$ and $\varphi_{\alpha'}$ are the polar and azimuthal angles of the measurement basis on the Bloch sphere.
We propose the use of the \textit{asymmetric Bell expectation} (ABE) score detailed in Ref.~\cite{sekatski2025certification} for BB84 state preparation, given by an average of the basis score $C=\frac{1}{m}\sum_{\alpha'=1}^mC_{\alpha'}$,
\begin{equation*}
    C_{\alpha'}=\frac{1}{2}\sum_{s,z\in\{0,1\}}(-1)^{s+z}[\cos\theta_{\alpha'}\mel{\psi^Z_s}{P^{\alpha'}_z}{\psi^Z_s}+\sin\theta_{\alpha'}\mel{\psi^X_s}{P^{\alpha'}_z}{\psi^X_s}],
\end{equation*}
where $\left\{P^{\alpha'}_z= \ketbra{\psi_z^{\alpha'}}{\psi_z^{\alpha'}}\right\}_z $ denote the prover's measurement operators in basis $\alpha'$.

Here, we generalize the ABE score to the six-state preparation,
\begin{equation}
\begin{split}
    C_{\alpha'}=&\frac{1}{2}\sum_{s,z\in\{0,1\}}(-1)^{s+z} \left[\cos\theta_{\alpha'}\mel{\psi^Z_s}{P^{\alpha'}_z}{\psi^Z_s}+\sin\theta_{\alpha'}\cos\varphi_{\alpha'}\mel{\psi^X_s}{P^{\alpha'}_z}{\psi^X_s}\right.\\
    &\left.+\sin\theta_{\alpha'}\sin\varphi_{\alpha'}\mel{\psi^Y_s}{P^{\alpha'}_z}{\psi^Y_s} \right],
\end{split}
\end{equation}
where we note that $C_{\alpha'}\leq\eta$. 
We can also consider the normalized version of the ABE score, denoted by $\cnorm$
\begin{equation}
    \cnorm = \frac{1}{\eta} C.
\end{equation}
In the protocol, the verifiers monitor both the transmission $\eta$ and the normalized score $\cnorm$, and compare them to pre-determined thresholds, $\etatol$ and $\ctolnorm$, respectively.

\begin{figure*}
\begin{protocol}{QPV with multiple measurement basis}
\begin{enumerate}
    \item For $i=1,\cdots,N$, repeat the following process:
    \begin{enumerate}
        \item \textbf{Verifier Preparation}: At the start of round $i$, the first and second verifiers randomly select ${x_i}\in[2^n]$ and ${y_i}\in[2^n]$ respectively. Verifiers also communicate and agree on a time $t_{P,i}$.
        \item \textbf{Quantum State Transfer}: A third party (or first verifier) prepares a quantum system $\ket{\psi^{{\alpha}}_s}_Q$, which corresponds to one of the six states, $\{\ket{0},\ket{1},\ket{+},\ket{-},\ket{+i},\ket{-i}\}$ with basis $\alpha=X,Y,Z$ and bit value $s\in[0,1]$. The third party sends the quantum system to the prover.
        \item \textbf{Verifier Message}: $V_1$ and $V_2$ send ${x_i}$ and ${y_i}$ respectively at $t_{P,i}-d_{V_1P}/c$ and $t_{P,i}-d_{V_2P}/c$ such that they both arrive at $P$ at time $t_{P,i}$.
        \item \textbf{Quantum Measurement}: The prover computes ${\alpha'_i}=f(x_i,y_i)$ and measures quantum system $Q$ in the computed basis. If no photon is detected, it sets $z_i=\perp$. The prover sends the measurement outcome $z_{1,i}=z_{2,i}=z_i$ to the verifiers.
        \item \textbf{Timing and Validity Check}: $V_1$ and $V_2$ record the arrival time of the respective responses $z_{1,i}$ and $z_{2,i}$ as $t_{z_1,i}$ and $t_{z_2,i}$. The verifiers first check if the responses are valid, i.e. (1) $z_{1,i}=z_{2,i}$. The verifiers then check if the timings are valid, (2) $t_{z_1,i}-t_{P,i}\leq d_{V_1P}/c+t_{\delta}$, and (3) $t_{z_2,i}-t_{P,i}\leq d_{V_2P}/c+t_{\delta}$ for some threshold $t_{\delta}$. If any checks fail, the protocol aborts immediately. 
    \end{enumerate}
    \item \textbf{Score Estimation}: $V_1$ and $V_2$ compute the transmission  $\eta$
    \begin{equation*}
        \eta = \frac{N_{\det}}{N}
    \end{equation*}
    and the normalized ABE score $\cnorm$:
        \begin{equation*}
            \cnorm = \frac{1}{m} \frac{1}{N_{\det}}\sum_{\alpha'} \left(\cos \theta_{\alpha'} E_{Z\alpha'} + \sin \theta_{\alpha'} \cos \varphi_{\alpha'} E_{X\alpha'} + \sin \theta_{\alpha'} \sin \varphi_{\alpha'} E_{Y\alpha'}\right),
        \end{equation*}
        where
        \begin{equation*}
            N_{\det} = \sum_{\alpha, \alpha'} \left(N_{\alpha \alpha',00} + N_{\alpha \alpha',11} + N_{\alpha \alpha',01} + N_{\alpha \alpha',10}\right)
        \end{equation*}
        is the number of detected rounds,
        \begin{equation*}
            E_{\alpha \alpha'} = N_{\alpha \alpha',00} + N_{\alpha \alpha',11} - N_{\alpha \alpha',01} - N_{\alpha \alpha',10},
        \end{equation*}
    and $N_{\alpha\alpha',sz}$ is the number of rounds where the verifier chooses basis $\alpha$ and bit value $s$, and the prover's measurement basis is set to $\alpha'$. The verifiers check if the transmission and the ABE score exceed the thresholds $\eta > \etatol$, and $\cnorm > \ctolnorm $. If both thresholds are exceeded, the prover's location is verified.
\end{enumerate}
\end{protocol}
\end{figure*}

To understand the significance of the ABE score, we consider an equivalent entanglement-based implementation of the protocol.
In the entanglement-based protocol, instead of preparing one of the six quantum states, the verifier prepares the maximally entangled state $\ket{\Phi^+}_{VQ} = \left(\ket{00}_{VQ} + \ket{11}_{VQ} \right)/\sqrt{2}$, and measures the quantum register $V$ in one of the three bases $\{X, Y, Z\}$. Let $\vec{\sigma}= (\sigma_X, \sigma_Y, \sigma_Z)$ denote the usual Pauli operators and let $\sigma_{P_{\alpha'}} := P_0^{\alpha'} - P_1^{\alpha'}$ denote the Hermitian operator corresponding to the prover's $\alpha'$-basis measurement. One can verify that
\begin{equation}
    C_{\alpha'} = \bra{\Phi^+}(\vec{n}_{\alpha'} \cdot \vec{\sigma}) \otimes \sigma_{P_{\alpha'}}\ket{\Phi^+},
\end{equation}
where $\vec{n}_{\alpha'} = (\sin \theta_{\alpha'} \cos \varphi_{\alpha'}, \sin \theta_{\alpha'} \sin \varphi_{\alpha'}, \cos \theta_{\alpha'})$ is a unit vector in the Bloch sphere. Therefore, $C_{\alpha'}$ can be interpreted as the correlation when both the verifier and the prover perform the projective measurement in the direction $\vec{n}_{\alpha'}$ of the Bloch sphere. It is well known that this is related to the error rate in the $\alpha'$-basis, $e_{\alpha'} \eta_{\alpha'} = (\eta_{\alpha'} - C_{\alpha'}) / 2$ if the verifier had measured in this basis (or, equivalently, prepared the corresponding states in the prepare-and-measure protocol). Consequently, if the transmission is basis-independent (which is the case for most implementations), the normalized ABE score $\cnorm$ will be related to the average error rate $e$: $\cnorm = 1 - 2e$.

Our protocol offers three key advantages.
Firstly, it reduces the complexity of experimental implementation by requiring fewer prepared states. 
Secondly, it provides greater flexibility since the preparation and measurement bases need not match. 
This allows the protocol performance to be adapted for different provers (e.g., in the context of QPV over a large quantum network), with provers further from verifier 1 utilizing more bases, without modifying the verifier's quantum device.
Finally, the verifiers no longer have to communicate details of the basis choice (e.g. by exchanging $x$ or $y$ via a secure communication channel) at the start of each round or before the QPV protocol.
This can reduce requirements on the verifiers since an authenticated communication channel is sufficient, eliminating the need for a secure communication channel.

\subsection{Security against Unentangled Adversaries}
In our QPV protocol, the adversaries' goal is to convince the verifiers that they are at location $P$ while not being present there. In other words, Alice and Bob aim to achieve a high transmission and normalized ABE score, i.e., $\eta > \etatol$ and $\cnorm > \ctolnorm$. Since our protocol only modifies the choice of state preparation and measurement but not the overall steps, the general adversary strategy matches that described earlier. Here, we consider how well they can achieve this task, provided that they do not pre-share an entangled state. While this analysis is only valid for a restricted adversarial model, it is still valuable since it allows us to understand the limits of the noise and loss tolerance of the protocol. In later sections, we consider more general adversarial models where the adversaries are allowed to pre-share quantum states of bounded dimensions, which may be entangled.

For simplicity, we assume that the evaluation function $f$ produces the uniform distribution over $[m]$. When adversaries share no entanglement, we consider a special case of that attack where the pre-shared quantum state $\rho_{AB}$ is separable, $\rho_{AB} = \sum_{r} p_{r} \rho_{A}^{r} \otimes \rho_{B}^{r}$ for some states $\{\rho_A^{r}\}_{r}$ and $\{\rho_B^{r}\}_{r}$, where $r$ represents pre-shared randomness.
We note that this attack model is independent of the speed of transmission of quantum system $Q$, so without any loss of generality, we can consider the system $Q$ to be sent at the speed of light. 

Let us consider a stronger adversary, who is provided with $x$ before the protocol round.
This fixes the optimal choice for Alice's local attack $\vars{E}_{QA}$. On the other hand, when Bob receives the input $y$, he can apply his local attack $\vars{E}_y$ to his state, and obtain the state $\sigma_{B}^{r,y} = \vars{E}_B^y[\rho_{B}^r]$. Then, following the adversary model described earlier, Alice and Bob will perform the communication and measurement step of the attack. Evidently, this is the most general attack that Alice and Bob can perform when they do not pre-share any entangled states.

However, notice that since $\rho_{AB}$ is separable, it is possible that Alice prepares the $2^n$ copies of Bob's initial states in her location where each copy is associated to one possible value of $y$. Next for each copy, Alice can simulate Bob's operation and the proper partitioning of the sub-systems. We note that in this strategy, Alice's actions are independent of the actual input $y$ that is sent by the verifier $V_2$. By discarding the irrelevant copies, the adversaries will end up with the same shared state as they have in the original attack we described previously. Therefore, when considering unentangled adversaries, one can simplify the analysis to a scenario where only Alice prepares the initial state, applies the local operations, and distributes the sub-systems. In this scenario, the resulting quantum state $\rho_{VAB}$ will be independent of the input $y$. After that, verifier $V_2$ will broadcast the input $y$ to Alice and Bob and they perform the measurements $\{A_{z}^y\}_z$ and $\{B_z^y\}_z$. We emphasize that this attack is actually equivalent to the original attack where Bob is involved in the preparation of the pre-measured quantum state.

For the purpose of upper bounding the ABE score, we can simplify the analysis further by assuming that the adversaries can pick the optimal attack for each basis choice, i.e. for each basis choice $\alpha'$, pick the $y$ with the strategy $\vars{E}_B^y$, $\{A^y_{z}\}_z$ and $\{B^y_z\}_z$ that provides the highest score and label this strategy as $\vars{E}_B^{\alpha'}$, $\{A^{\alpha'}_{z}\}_z$ and $\{B^{\alpha'}_{z}\}_z$. Therefore, to get an upper bound of the adversaries' ABE score, it is sufficient to consider the reduced scenario where Alice applies the local operation independently of the inputs, sends some sub-system to Bob, and afterwards, the verifier $V_2$ announces the measurement basis $\alpha'$ to both adversaries, who can only perform local measurements to their respective sub-systems. Moreover, after reducing the problem to the simplified scenario, we can consider a larger Hilbert space where the local operations are unitary and the local measurements are projective. In this case, we label $\ket{\phi_s^{\alpha}} = \vars{E}_{QA} \ket{\psi_{s}^{\alpha}}$ as the state shared by Alice and Bob after the local operations, conditioned on the verifier's preparation basis choice and bit value.

We are interested in finding the maximum ABE score achievable by the adversary for a given transmission $\eta$.
If a prover achieves a higher score, we can conclude with high confidence that the prover is at location $P$.
With the constraints placed on the transmission and response mismatch rate from the adversarial model, the maximum achievable score can be upper bounded by an optimization problem
\begin{equation}
\begin{split}
    \max\quad&  \frac{1}{2m} \sum_{\alpha'}\sum_{s,z\in\{0,1\}} (-1)^{s+z} \left[ \cos\theta_{\alpha'} \mel{\phi_{s}^Z}{A^{\alpha'}_{z}}{\phi_{s}^Z} + \sin\theta_{\alpha'}\cos\varphi_{\alpha'} \mel{\phi_{s}^X}{A^{\alpha'}_{z}}{\phi_{s}^X}\right.\\
    &\left.+ \sin\theta_{\alpha'}\sin\varphi_{\alpha'} \mel{\phi_{s}^Y}{A^{\alpha'}_{z}}{\phi_{s}^Y}\right]\\
    \mathrm{s.t.}\quad&\frac{1}{6}\sum_{\alpha,s}\mel{\phi_{s}^{\alpha}}{A^{\alpha'}_{\perp}}{\phi_{s}^{\alpha}}=1-\eta,\quad\forall \alpha'\\
    &\mel{\phi_{s}^{\alpha}}{A^{\alpha'}_z\otimes B^{\alpha'}_{z'}}{\phi_{s}^{\alpha}}=0,\quad\forall z\neq    z',\alpha,\alpha',s\\
    & \braket{\phi_{s}^{\alpha}}{\phi_{\tilde{s}}^{\tilde{\alpha}}}=\braket{\psi^{\alpha}_s}{\psi^{\tilde{\alpha}}_{\tilde{s}}},\quad\forall\alpha,\tilde{\alpha},s,\tilde{s}
\end{split}
\end{equation}
The constraints correspond to (1) loss is $1-\eta$ and independent of $(x,y)$ (equivalently $\alpha'$), (2) classical response cannot have any mismatch (both adversaries cannot give different responses $z\neq z'$), and (3) the overlap between different preparation states is preserved under a fixed map $\vars{E}_{QA}$. We also note that due to the second constraint, we can ignore Bob's response and focus on Alice's response only when computing the score as well as the transmission.

We can use the relaxation introduced in Ref.~\cite{wang2019characterising}, which is based on the Navascu\'es-Pironio-Ac\'in (NPA) hierarchy~\cite{NPA2007,NPA2008}, to relax the optimization into a semidefinite program (SDP).
More concretely, we can construct successively larger sets of projective measurements: we start with the first set of projectors $\Gamma^1 = \{\mathbb{I}\} \cup \{A_{z}^{\alpha'}\}_{z \in \{0,1,\perp\}, \alpha' \in [m]} \cup \{B_{z}^{\alpha'}\}_{z \in \{0,1,\perp\}, \alpha' \in [m]}$, and for any higher level $\ell > 1$, we can inductively define the set $\Gamma^\ell = \{\gamma_i \gamma_j: \gamma_i \in \Gamma^{\ell - 1}, \gamma_j \in \Gamma^{1} \}$. Next, for a chosen level $\ell$, we consider the set of (sub-normalized) vectors $\{\ket{\xi_{rs}} = \gamma_r \ket{\phi_{s}^Z}: \gamma_r \in \Gamma^{\ell}\}_{rs}$ which has a corresponding Gram matrix $G^\ell$ whose entries are the inner-products $\braket{\xi_{rs}}{\xi_{\bar{r}\bar{s}}}$. We note that the Gram matrix $G^\ell$ is positive semidefinite by definition, $G ^\ell \succeq 0$. Note that without any loss of generality, we can consider the measurement operators to be projective measurements using Naimark's dilation theorem. We also note that it is sufficient to work with the $Z$-basis states since we can write the $X$-basis states as $\ket{\phi_{s}^X} = (\ket{\phi_{0}^Z} + (-1)^s \ket{\phi_{1}^Z})/\sqrt{2}$ and the $Y$-basis states as $\ket{\phi_{s}^Y}=(\ket{\phi_{0}^Z} + i (-1)^s \ket{\phi_{1}^Z})/\sqrt{2}$ due to the linearity of the adversaries' unitary operations. As such, for any operator $\Pi$, any expression of the form $\mel{\phi_{s_1}^{\alpha_1}}{\Pi}{\phi_{s_2}^{\alpha_2}}$ can be written as linear combinations of $\left\{\mel{\phi_{s}^Z}{\Pi}{\phi_{\tilde{s}}^Z} \right\}_{s, \tilde{s}}$.
The problem then can be relaxed to
\begin{equation}
\label{eqn:unentangled_adv_SDP}
\begin{split}
    \max\quad&\frac{1}{2m}\sum_{s,z\in\{0,1\},\alpha'}(-1)^{s+z}\left[\cos\theta_{\alpha'} \mel{\phi_{s}^Z}{A^{\alpha'}_{z}}{\phi_{s}^Z} + \sin\theta_{\alpha'}\cos\varphi_{\alpha'} \mel{\phi_{s}^X}{A^{\alpha'}_{z}}{\phi_{s}^X}\right.\\
    &\left.+ \sin\theta_{\alpha'}\sin\varphi_{\alpha'} \mel{\phi_{s}^Y}{A^{\alpha'}_{z}}{\phi_{s}^Y}\right]\\
    \mathrm{s.t.}\quad&\frac{1}{2}\sum_{s}\mel{\phi_{s}^{Z}}{A^{\alpha'}_{\perp}}{\phi_{s}^{Z}}=1-\eta,\quad\forall \alpha'\\
    &\mel{\phi_{s}^Z}{A^{\alpha'}_z B^{\alpha'}_{z'} }{\phi_{\tilde{s}}^Z}=0,\quad\forall z\neq z', \alpha', s, \tilde{s}\\
    & \braket{\phi_{s}^{Z}}{\phi_{\tilde{s}}^{Z}} = \delta_{s\tilde{s}}, \quad \forall s,\tilde{s}\\
    & [A_z^{\alpha'},B_{\tilde{z}}^{\tilde{\alpha'}}]=0,\quad\forall \alpha',\tilde{\alpha}',z,\tilde{z}\\
    &A_z^{\alpha'}A_{\tilde{z}}^{\alpha'}=\delta_{z\tilde{z}} A_{z}^{\alpha'},\quad B_z^{\alpha'}B_{\tilde{z}}^{\alpha'}=\delta_{z\tilde{z}} B_{z}^{\alpha'},\quad\sum_zA_z^{\alpha'}=\mathbb{I},\quad\sum_zB_z^{\alpha'}=\mathbb{I},\quad\forall \alpha'\\
    &G^\ell \geq 0,
\end{split}
\end{equation}
where the additional set of constraints include the commutativity of $A$ and $B$ operators and their projective nature.
Since the problem can be expressed as elements of the positive semi-definite matrix $G$, the problem is an SDP and can be solved numerically.
Since the SDP is a relaxation of the original optimization problem, the solution is a valid upper bound to the maximum score that an adversary can achieve.

We solve the SDP numerically using ncpol2sdpa~\cite{Ncpol2sdpa} with the picos interface~\cite{cvxpy} and qics solver~\cite{qics} at NPA level 1+AB (using projectors $\Gamma^1\cup\{A_z^{\alpha'}B_{\tilde{z}}^{\tilde{\alpha}'}\}_{\alpha'\tilde{\alpha}'z\tilde{z}}$). To compare with previous error rate-based analysis, we convert the maximum achievable score of the adversaries to a minimum error rate (recall $C=\eta(1-2e)$, which is achievable for an honest party with average transmission of $\eta$ and QBER $e$).
Figure~3 shows the numerical results, for three and four measurement basis choices, for both the simple and original multi-basis QPV, and comparing our SDP in Eq.~\eqref{eqn:unentangled_adv_SDP} and that of Ref.~\cite{Llorenc2023}.
As argued earlier, the simple and original multi-basis QPV rely on equivalent sets of statistics to demonstrate security, and thus as expected, their performance matches (same solid line).

For different numbers of measurement bases, the behaviors are similar: the adversaries' minimum error rate decreases with decreasing transmission $\eta$, and depending on the number of prover's measurement bases $m$, the minimum error rate goes to zero when the transmission is approximately $1/m$. We note that this behavior is the same as the one expected from the case where $m$ preparation and measurement bases are used and the error rate is monitored directly. This gives numerical evidence that the ABE score (or its normalized version) is a good substitute to the error rate, to monitor the security of the QPV protocol.

Figure~3 shows that our SDP formulation improves the protocol's performance compared to the method used in Ref.~\cite{Llorenc2023}. This is because our SDP uses the characterization of the state preparation explicitly, whereas Ref.~\cite{Llorenc2023} worked with a relaxation of this constraint (see \textbf{Proposition 4.2}). 

\begin{figure}
    \centering
    \includegraphics[width=0.8\linewidth]{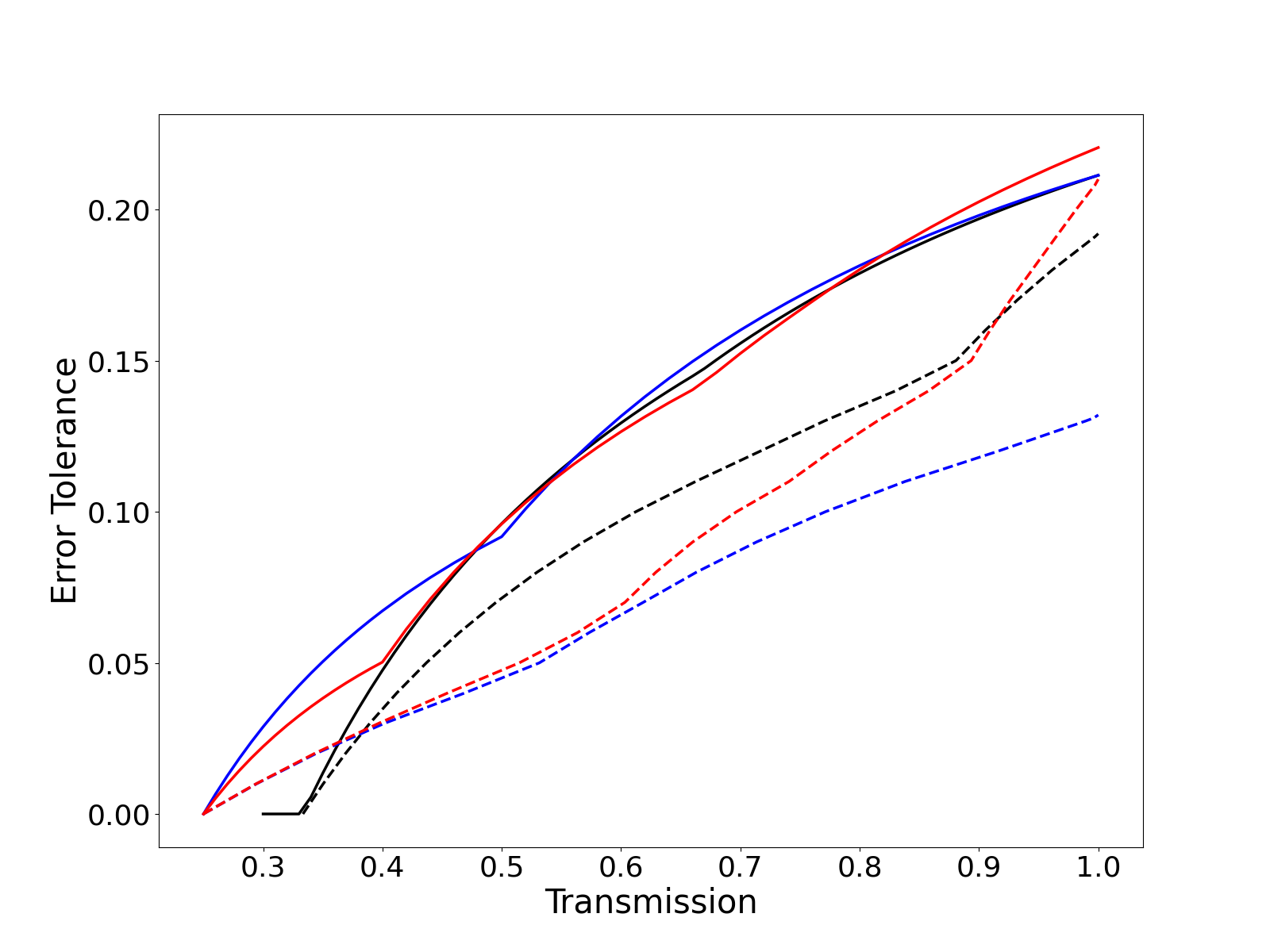}
    \caption{Plot of error tolerance at various transmission $\eta$ for secure position verification against unentangled adversaries. The solid lines represent the simple multi-basis QPV analyzed using the SDP in Eq.~\eqref{eqn:unentangled_adv_SDP} at level 1+AB, and it matches exactly the original multi-basis QPV in Ref.~\cite{Llorenc2023}, using an analogous SDP. This is an improvement over the original multi-basis QPV analyzed using their SDP in Ref.~\cite{Llorenc2023} at level 2, represented by the dotted lines. The results are shown for the 3 basis case (basis in $X$, $Y$ and $Z$) [black lines], 4 basis case with basis in a tetrahedron arrangement ($(\theta_{\alpha'},\varphi_{\alpha'})\in\{(0,0)\}\cup\{(\cos^{-1}(-1/3),\frac{2k\pi}{3})\}_{k=1,2,3}$) [blue lines] and 4 basis case with $(4,2,3)$ arrangement~\cite{Llorenc2023} ($(\theta_{\alpha'},\varphi_{\alpha'})\in\{(0,0)\}\cup\{(\frac{\pi}{2},\frac{2k\pi}{3})\}_{k=1,2,3}$) [red lines].}
    \label{fig:Multi_Basis_QPV_Plot}
\end{figure}

\subsection{Security against Entangled Adversaries}
In the previous section, we presented the security analysis against unentangled adversaries, and demonstrated improvements in proof techniques that led to significant increases in error tolerance. Our technique also enables the use of the ABE score to monitor the security of the QPV protocol. We now extend this analysis to adversaries that may pre-share entangled quantum states.

More precisely, referring to the general attack model outlined earlier, we assume that the pre-shared state $\rho_{AB} = \ketbra{\psi}{\psi}_{AB}$ is pure and the dimension of each sub-system is bounded by $2^q$ for some fixed $q \geq 1$. Moreover, we assume that the local operations of each party are restricted to unitary operations, denoted as $\{U_{QA}^x\}_x$ and $\{U_B^y\}_y$. Since the quantum state sent by the verifier $V_1$ is a single-qubit, the dimension of the unitary $U_{QA}^x$ is $2^{q+1}$, while the dimension of the unitary $U_B^y$ is $2^q$. We call this strategy a $q$\textit{-qubit restricted strategy} to distinguish from a more general attack strategy introduced in a later section. We also remark that the $q$-qubit restricted strategy is essentially the same attack as the $q$-qubit pure strategy considered in Ref.~\cite{Llorenc2023}, except that the analysis in that work can only be applied to protocols that are based on error rate monitoring.

\begin{figure}
    \centering
    \includegraphics[width=\linewidth]{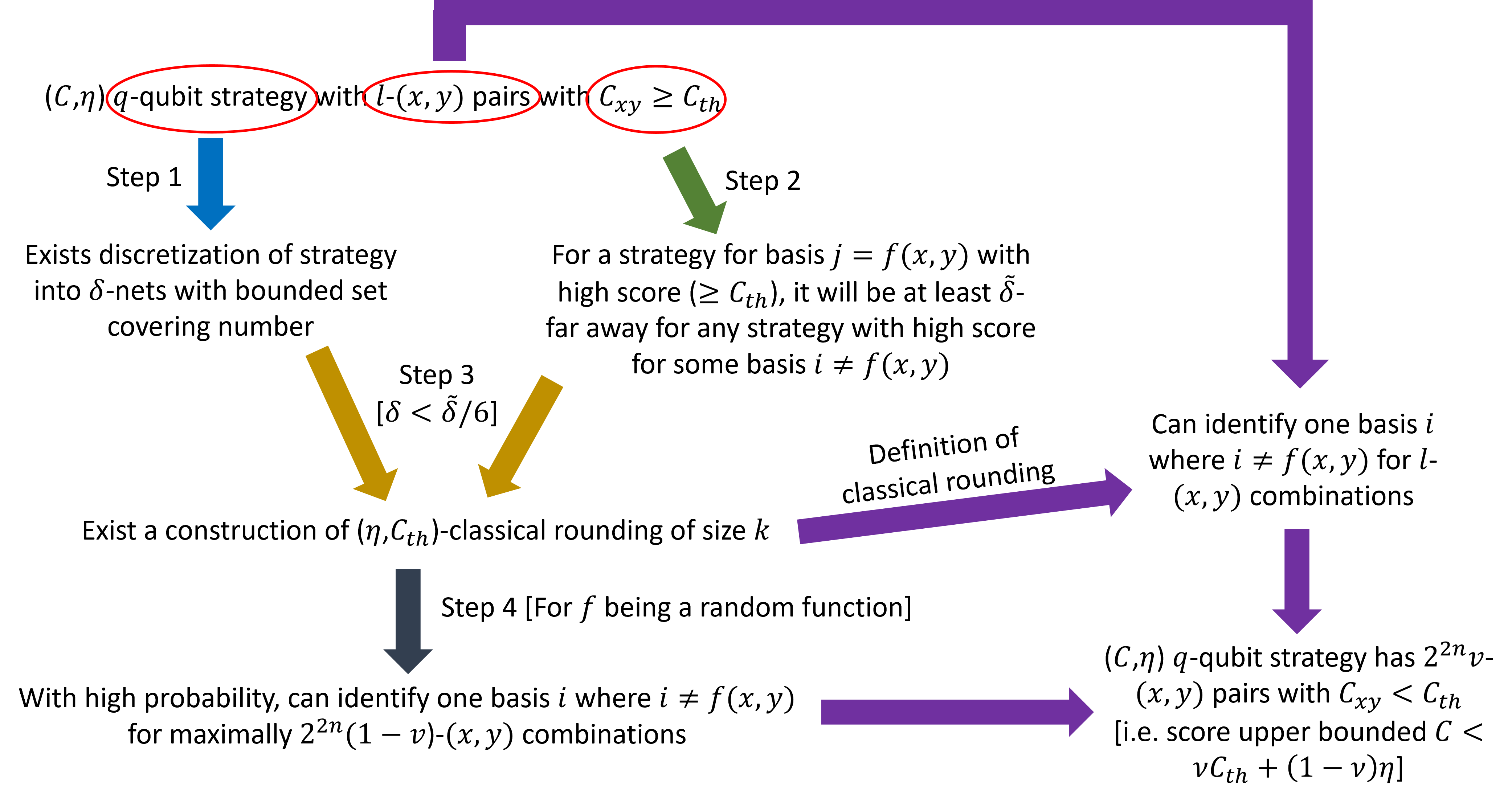}
    \caption{Summary of the five main proof steps in QPV security analysis against entangled adversaries.}
    \label{fig:QPV_Proof_Chain}
\end{figure}

To prove security, we closely follow the security proof in Ref.~\cite{Llorenc2023} but with three key modifications: (1) improvements to the classical rounding definition, (2) tightening of the trace distance bound, and (3) adaptation to use score instead of error rate.
The proof can be separated into five steps, summarized in Figure~4.
The first step is to characterize $q$-qubit strategies using parameters monitored in the protocol, namely the ABE score $C$ and the transmission $\eta$. This allows us to discretize the set of strategies into a $\delta$-net, where every strategy is $\delta$-close to the center of a net in some distance measure.
Without loss of generality, we can consider an equivalent entanglement-based protocol where the first verifier prepares the Bell state $\ket{\Phi^+}_{VQ}$ and performs a measurement on subsystem $V$ later, i.e. $\ket{\psi^{xy}_{\alpha s}}_{A'B'}$ is formed from the post-measurement state of $\ket{\Psi^{xy}}=(\mathbb{I}_V \otimes U_{QA}^x \otimes U_B^y)(\ket{\Phi^+}_{VQ} \otimes \ket{\psi}_{AB})$ for measurement operator $\Pi^{\alpha}_s$ (note for $Y$ basis, this includes a bit flip to have the correct correlation).
The following Hermitian operator is very useful when calculating the ABE score
\begin{equation*}
    \Lambda^{\alpha'}_s=\cos\theta_{\alpha'}\Pi^{Z}_s+\sin\theta_{\alpha'}\cos\varphi_{\alpha'}\Pi^{X}_s+\sin\theta_{\alpha'}\sin\varphi_{\alpha'}\Pi^{Y}_s.
\end{equation*}

We characterize a $q$-qubit restricted strategy using two parameters, namely the ABE score $C$ and the transmission $\eta$. Thus, we define a restricted $q$-qubit $(C,\eta)$-strategy as
\begin{definition}[Restricted $q$-qubit $(C,\eta)$-strategy]
    For any restricted $q$-qubit strategy, the input-dependent score and transmission are defined as
    \begin{gather*}
        C_{xy}= \sum_{s,z\in\{0,1\}}(-1)^{s+z} \mel{\Psi^{xy}}{\Lambda^{f(x,y)}_s\otimes A^{xy}_{z}\otimes B_{z}^{xy}}{\Psi^{xy}}\\
        p_{t|xy}=1-\mel{\Psi^{xy}}{\mathbb{I}_V\otimes A^{xy}_{\perp}\otimes B_{\perp}^{xy}}{\Psi^{xy}}
    \end{gather*}
    respectively, and the matching condition is defined as
    \begin{equation*}
        \mel{\Psi^{xy}}{\mathbb{I}_V\otimes A^{xy}_{z}\otimes B_{z'}^{xy}}{\Psi^{xy}}=0,\,\forall z\neq z'.
    \end{equation*}
    A $q$-qubit restricted strategy is a $(C,\eta)$-strategy if it satisfies $2^{-2n}\sum_{xy}C_{xy}=C$, $p_{t|xy}=\eta$ and the matching condition for all $(x,y)$. Furthermore, the strategy has a $(\cth,l)$-partition if there are at least $l$ pairs of $(x,y)$ satisfying $C_{xy}\geq \cth$.
\end{definition}
Using set covering arguments~\cite{Vershynin2012}, we construct $\delta$-nets for the pure states and unitaries of Alice and Bob respectively (see Theorems~\ref{thm:covering_number_state} and \ref{thm:covering_number_unitary}), with the number of nets $\log_2\abs{\vars{N}_S}\leq 2^{2q+1}\log_2(1+\frac{2}{\delta})$ corresponding to the pure state with dimension $d=2^{2q}$, $\log_2\abs{\vars{N}_U}\leq 2^{2q+3}\log_2(1+\frac{2}{\delta})$ corresponding to Alice's unitary of dimension $d=2^{q+1}$ and $\log_2\abs{\vars{N}_U}\leq 2^{2q+1}\log_2(1+\frac{2}{\delta})$ corresponding to Bob's unitary of dimension $d=2^q$.
We note that the definition of threshold values $\cth$ is simply to aid in the security analysis and is not tied to the strategy.
Their choices partition a strategy's attacks into groups, similar to the $(\varepsilon,l)$-perfect strategy definition in Ref.~\cite{Llorenc2023}, where the $(x,y)$-pairs are partitioned into two groups based on the error rate.
Since their choices do not impact security in practice, we are free to choose the threshold $\cth$.

The second step considers the same or an optimal set of strategies satisfying some threshold score $\cth$ for fixed $(x,y)$, and shows that strategies corresponding to different measurement basis $\alpha'=f(x,y)$ must be well-separated, with trace distance at least $\tilde{\delta}$.
For the ABE score, we define the optimal set of strategies for fixed $(x,y)$ as
\begin{definition} \label{def: C-eta set}
    Let $\tilde{C}\in[-1,1]$ and $\tilde{\eta}\in[0,1]$, and define the $\alpha'$-contribution to the score, transmission and abort probability for a given choice of measurement operators $\{A_z\}_z$ and $\{B_z\}_z$ as
    \begin{gather*}
        C_{\alpha'}=\sum_{s,z\in\{0,1\}}(-1)^{s+z}\mel{\Psi}{\Lambda^{\alpha'}_s\otimes A_{z}\otimes B_{z}}{\Psi}\\
        p_t=1-\mel{\Psi}{\mathbb{I}_V\otimes A_{\perp}\otimes B_{\perp}}{\Psi}\\
        p_a=\sum_{z\neq z'}\mel{\Psi}{\mathbb{I}_V\otimes A_z\otimes B_{z'}}{\Psi}
    \end{gather*}
    The set of output states for a partition for a basis $\alpha'$ is defined as
    \begin{gather*}
        \vars{S}_{{\alpha'}}^{\tilde{C},\tilde{\eta}}=\{\ket{\Psi}_{V A'B'}=(\mathbb{I}_V \otimes U_{QA}\otimes U_B)(\ket{\Phi^+}_{VQ} \otimes \ket{\psi}_{AB}):\exists \{A_z\}_z,\{B_z\}_z,\, s.t.,\,C_{\alpha'}\geq\tilde{C},\,p_t= \tilde{\eta},p_a=0\},
    \end{gather*}
    where $U_{QA}$, $U_B$ and $\ket{\psi}_{AB}$ are restricted to be of a $q$-qubit restricted strategy.
\end{definition}
We seek to lower bound the trace distance between strategies corresponding to different basis $\alpha'=i,j$ ($i\neq j$), $\Delta(\ket{\Psi_i},\ket{\Psi_j})\geq\tilde{\delta}$, for $\ket{\Psi_i}\in\vars{S}_i^{\cth,\eta}$ and $\ket{\Psi_j}\in\vars{S}_j^{\cth,\eta}$.
Previous approaches lower bound $\tilde{\delta}$ by either a reduction to a simpler QPV game~\cite{Llorenc2023} or using Fano's inequality, complementary-information trade-off and Alicki-Fannes-Winter continuity bound~\cite{Bluhm2022}.

Here, we obtain a tightened trace distance bound by directly minimizing the trace distance over all combinations of states from the two sets.
We first perform a piecewise linear approximation of the trace distance, $\Delta(\ket{\Psi_i},\ket{\Psi_j})=\sqrt{1-\abs{\braket{\Psi_i}{\Psi_j}}^2}$, with respect to the inner product $\abs{\braket{\Psi_i}{\Psi_j}}$ that lower bounds the trace distance value for $\ket{\Psi_i}\in \vars{S}_i^{\cth,\eta}$.
Since the global phase of the state impacts neither the trace distance nor statistics, WLOG we take $\braket{\Psi_i}{\Psi_j}$ to be real and positive.
Minimizing a linear function of $\braket{\Psi_i}{\Psi_j}$ subject to the constraints on the expectation values, we can use the relaxation in Ref.~\cite{wang2019characterising} based on the NPA hierarchy to relax the optimization into an SDP.
The new trace distance result can be summarized in the corollary below, and details can be found in the Methods section.
\begin{corollary}\label{cor: existence trace distance bound}
    Let $j \in [m]$. For any state $\ket{\Psi_j} \in \vars{S}_j^{\cth, \eta}$, there exists $i \in [m]$ such that $i \neq j$ and for any state $\ket{\Psi_i} \in \vars{S}_i^{\cth, \eta}$, we have $\Delta(\ket{\Psi_i}, \ket{\Psi_j})\geq \tilde{\delta}$, where $\tilde{\delta}$ is the solution to the optimization problem given by Eq.~\eqref{eq: sdp six-state}.
\end{corollary}
We note that while the trace distance lower bound here applies to sets $\vars{S}_{\alpha'}^{\cth,\eta}$ with no dimension restrictions, it remains valid when considering only $q$-qubit restricted strategies since these form a subset of $\vars{S}_{\alpha'}^{\cth,\eta}$.

The third step involves demonstrating that by choosing a small enough $\delta$-net, we can assign each net to a strategy.
This allows the construction of a classical rounding strategy, where messages $x$ and $y$ can be compressed and still recover $f(x,y)$, which we formally define as
\begin{definition}[Classical Rounding, modified from Ref.~\cite{Bluhm2022,Llorenc2023}]
    A function $g:\{0,1\}^{3k}\rightarrow\{\vars{S}:\vars{S}\subseteq[m],\abs{\vars{S}}=m-1\}$ is termed a $(\eta,\cth)$-classical rounding of size $k$ if for any function choice $f:\{0,1\}^{2n}\rightarrow[m]$, any $l\in[2^{2n}]$, any $(C,\eta)$-strategy with $(\cth,l)$-partition, there exist functions $f_A:\{0,1\}^n\rightarrow\{0,1\}^k$, $f_B:\{0,1\}^n\rightarrow\{0,1\}^k$ and $\lambda\in\{0,1\}^k$ such that $f(x,y)\in g(f_A(x),f_B(y),\lambda)$ for $l$ pairs of $(x,y)$.
\end{definition}
Here, we make a slight tweak to the classical rounding definition by requiring $\abs{\vars{S}}=m-1$, which tightens the analysis.
Based on this definition, we can show that one can construct a classical rounding by choosing appropriate $\cth$ and $\eta$ values such that Corollary~\ref{cor: existence trace distance bound} (separation of $q$-qubit restricted strategies) is valid,
\begin{theorem}
\label{thm:rounding_from_delta_sec}
    Suppose $\cth$ and $\eta$ are such that $\tilde{\delta}(\cth,\eta)>0$. Then, there exists a $(\eta,\cth)$-classical rounding of size $k=2^{2q+3}\log_2\Big\lceil2+\frac{12}{\tilde{\delta}}\Big\rceil$.
\end{theorem}
\begin{proof}
    See Methods.
\end{proof}

The fourth step is to argue that for $q$-qubit strategies, if there exists a classical rounding, when a random function $f$ is used, then except with small probability, there is a maximum number of $(x,y)$ pairs that can be correctly classified by the classical rounding. 
We summarize this step in a theorem,
\begin{theorem}
\label{thm:nu_value_compute}
    Fix a $(\eta,\cth)$-classical rounding with $k=2^{2q+3}\log_2\Big\lceil2+\frac{12}{\tilde{\delta}}\Big\rceil$ and let $q\leq\frac{1}{2}n-q_0$. Then, a uniform random function $f:\{0,1\}^{2n}\rightarrow[m]$ fulfills the following with probability at least $1-2^{-\beta}$: For any $f_A$, $f_B$ and $\lambda$, $f(x,y)\in g(f_A(x),f_B(y),\lambda)$ holds for less than $2^{2n}(1-\nu)$ pairs of $(x,y)$, for 
    \begin{equation*}
        \nu= h_b^{-1}\left\{\log_2\left(\frac{m}{m-1}\right)-2^{5-2q_0}\log_2\Bigg\lceil2+\frac{12}{\tilde{\delta}}\Bigg\rceil-\frac{\beta}{2^{2n}}\right\},
    \end{equation*}
    where $h_b(x)$ is the binary entropy function.
\end{theorem}
\begin{proof}
See Methods.
\end{proof}

Finally, by the definition of classical rounding, a strategy with $(C_{th},l)$-partition has $l$ pairs of $(x,y)$ that can be correctly classified.
Coupled with Theorem~\ref{thm:nu_value_compute}, this implies a minimum number of pairs of $(x,y)$ where the score falls below the score threshold.
We can therefore conclude the proof via contradiction that if the score exceeds a given threshold and the memory size is bounded, no such $q$-qubit restricted strategy can exist and thus with high probability, the response must be from a party at $P$. 
Formally, we are able to demonstrate that any $q$-qubit restricted strategy would result in an upper bound on the score $C$ in a $(C,\eta)$-strategy.
\begin{theorem}
    Let $f$ be a random function, and let Alice and Bob be restricted to $q$-qubit restricted strategies with $q\leq\frac{1}{2}n-q_0$. For any $(C,\eta)$-strategy with $(\cth,\eta)$-partition with $\cth$ chosen such that $\tilde{\delta}(\cth,\eta)>0$, with probability $1-2^{-\beta}$, the score is upper bounded by $C\leq (1-\nu)\eta+\nu \cth$.
\end{theorem}
\begin{proof}
From Theorem~\ref{thm:rounding_from_delta_sec}, an appropriate choice of $\cth$ allows the formulation of a $(\eta,\cth)$-classical rounding $g$.
From Theorem~\ref{thm:nu_value_compute}, this classical rounding, coupled with a bound on $q$, implies that with probability of at least $1-2^{-\beta}$, $f(x,y)\in g(f_A(x),f_B(y),\lambda)$ holds for less than $2^{2n}(1-\nu)$ pairs of $(x,y)$.
Suppose there exists a $(\cth,l)$-partition for the $(C,\eta)$-strategy.
Using the same classical rounding, we can find $f_A$, $f_B$ and $\lambda$ such that $f(x,y)\in g(f_A(x),f_B(y),\lambda)$ holds for $l$ pairs of $(x,y)$.
If $l>2^{2n}(1-\nu)$, we have a contradiction.
Therefore, a $(C,\eta)$-strategy can have a $(\cth,l)$-partition with at most $l\leq 2^{2n}(1-\nu)$.
Noting that the ABE score is upper bounded by $\eta$, we can upper bound the average score,
\begin{equation}
\label{eq:Entangled_CUB_form}
    C=2^{-2n}\sum_{xy}C_{xy}\leq2^{-2n}\left(\sum_{xy:C_{xy}< \cth}\cth +\sum_{xy:C_{xy}\geq \cth}\eta
    \right)\leq\nu \cth +(1-\nu)\eta.
\end{equation}
\end{proof}
Since we are free to choose $\cth$, by minimizing $C$ over the choice of $\cth$, we can compute an upper bound $C^{UB}$ for any $q$-qubit restricted strategy with transmission $\eta$.
Therefore, if a prover is able to achieve a score $C>C^{UB}$ in the asymptotic regime (no statistical fluctuations), we can be confident (except with probability $\leq 2^{-\beta}$) that there must be a party present at $P$ participating in the QPV protocol.

The main improvements made to the analysis are: (1) generalizing to a score-based analysis, (2) tightening the trace distance lower bound, and (3) tightening Theorem~\ref{thm:nu_value_compute} via a modified classical rounding definition.
The first improvement is a generalization which does not result in performance improvements, while the latter two improve the security via the value of $\nu$.
In the improved analysis,
\begin{equation}
    \log_2\left(\frac{m}{m-1}\right)-h_b(\nu)-2^{5-2q_0}\log_2\bigg\lceil2+\frac{12}{\tilde{\delta}}\bigg\rceil=\frac{\beta}{2^{2n}}
\end{equation}
determines $\nu$.
Compared to Ref.~\cite{Llorenc2023}, we have a tighter (larger) trace distance bound $\tilde{\delta}$ and the absence of a $-\nu\log_2(m-1)$ on the LHS of the equation.
Both improvements lead to a higher $\nu$ for fixed $q_0$ and $\frac{\alpha}{2^{2n}}$, which in turn leads to a lower upper bound on the score (or equivalently a higher upper bound on the error).

To illustrate the improvement, we compare our analysis to that of Ref.~\cite{Llorenc2023}. 
For the latter, we use an SDP analogous to Eq.~\ref{eqn:unentangled_adv_SDP} for error rate when computing the winning probability $w^{\xi}(\eta)$.
For simplicity, we select $\frac{\beta}{2^{2n}}=10^{-10}$ and optimize the score upper bound $C^{UB}$ over choices of $\cth$ at each $\eta$ value.
We convert the maximum score to the minimum error rate using $C=\eta(1-2e)$ for comparison with error rate-based analysis.
Figure~5 shows that for fixed $q_0$, our improvements increase the error tolerance, thereby relaxing the experimental requirements.
The choice of $q_0$ and thus the size of the adversary's quantum systems can also heavily influence the performance, though this can be compensated by increasing the size of the classical messages $x$ and $y$ sent.

\begin{figure}
    \centering
    \includegraphics[width=0.8\linewidth]{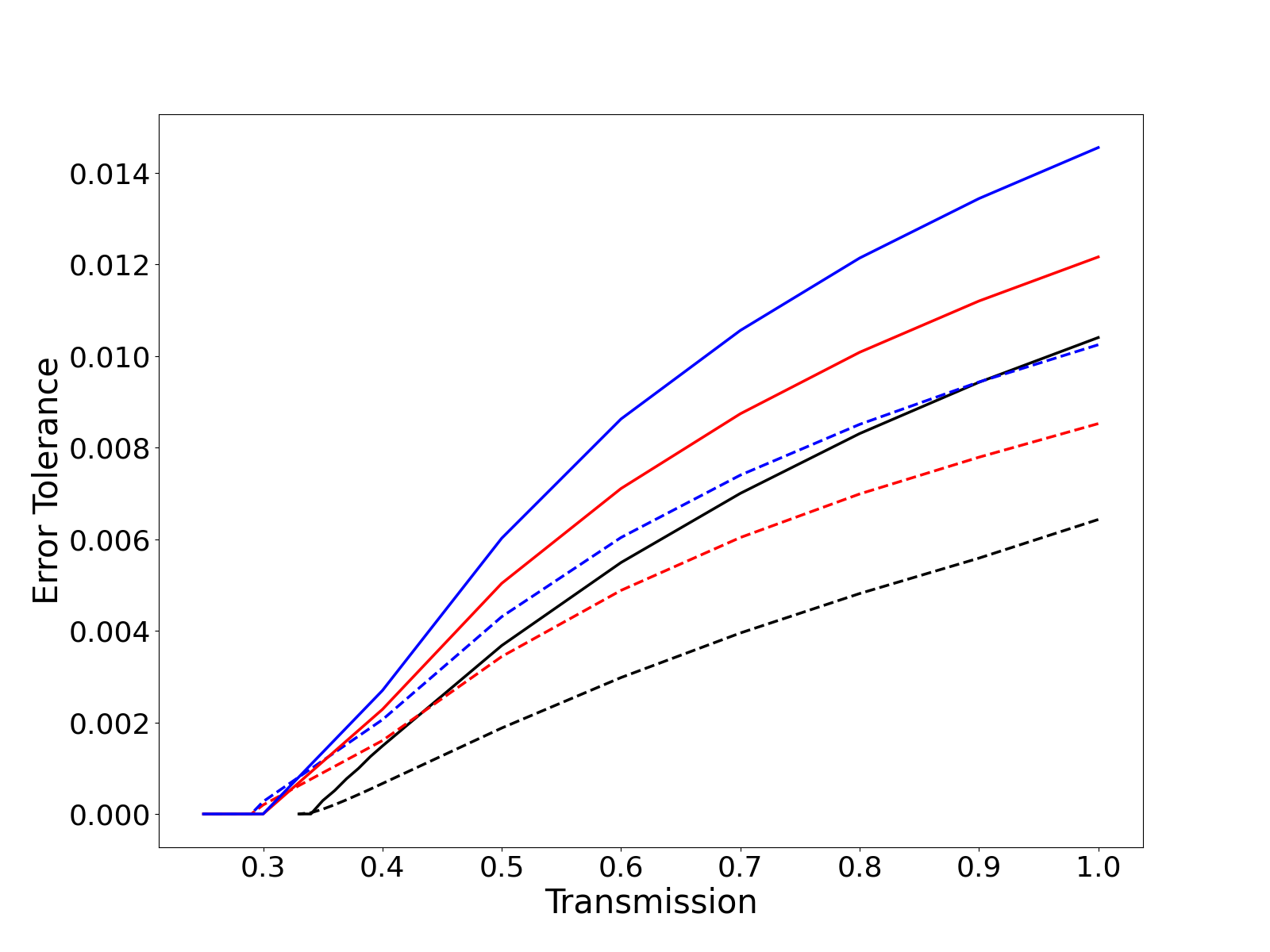}
    \caption{Plot of error tolerance at various transmission for secure position verification against an entangled adversary. The simulation is performed for 3 basis (6 states) [black lines] and for 4 basis arranged in a tetrahedron for $q_0=5$ [red lines] and $q_0=5.5$ [blue lines]. The results of the improved analysis (solid lines) are compared to the analysis in Ref.~\cite{Llorenc2023}, using a modified SDP (dotted lines).}
    \label{fig:entangled_adv_plot}
\end{figure}

We can study fundamental limits of the security by studying the asymptotic behavior, where we choose $q_0=\frac{n}{3}$ and let $n\rightarrow\infty$.
In this scenario, if the trace distance $\tilde{\delta}>0$, $\nu$ depends solely on $h_b(\nu)=\log_2(\frac{m}{m-1})$.
From Eq.~\eqref{eq:Entangled_CUB_form}, the error rate required for a secure QPV against entangled adversaries can be expressed as a factor $\nu$ of that against unentangled adversaries, i.e. $e_{ent}\leq \nu e_{unent}$.
This implies that there is a fundamental upper limit on the error rate given by $h_b^{-1}(\log_2(\frac{m}{m-1}))$, which decreases with $m$. This improves upon Ref.~\cite{Llorenc2023}, which has an additional $\nu\log_2(m-1)$ penalty in the asymptotic limit.
Since $\nu\leq\frac{1}{2}$, this method of security analysis yields an error rate that is at most half that against unentangled adversaries.

\subsection{Generalization of Security against Entangled Adversaries}
While the $q$-qubit restricted strategy accounts for adversaries that pre-share entangled states, there are still some assumptions that may not hold in general. To cover the most general attack for adversaries that are only limited by the Hilbert space dimension of their pre-shared quantum states, we need to allow the adversaries to: (1) Use shared randomness $R$ with an unbounded dimension, (2) Pre-share mixed states $\rho_{AB}$ instead of pure states, (3) Use transmission that depends on the inputs $(x,y)$ and the shared randomness $r$, (4) Apply general quantum channels (i.e., CPTP maps) $\channel_{AQ}^{xr}$ and $\channel_B^{yr}$ instead of assuming that the local operations are unitaries.

In this section, we present a framework that fully addresses the first three points, and partially addresses the fourth. More specifically, we consider a class of CPTP maps that takes a quantum system of dimension $2^{q+1}$ and outputs a classical-quantum system whose combined dimension is at most $2^{q+1}$ for Alice (and $2^q$ for Bob). In contrast to $q$-qubit restricted strategies, we refer to this class of strategies as \textit{$q$-qubit general strategies}. In the most general case, $q$-qubit entangled adversaries should be able to perform local operations that output classical registers of arbitrary dimensions, and the dimension bound should only be applied to the quantum registers. We leave this final generalization for future work.

Generalizing the adversarial model from the $q$-qubit restricted strategy is important because the restriction to restricted strategy is easily violated -- it is trivial to perform the more general strategy instead of the restricted strategy. Moreover, it might be advantageous for the adversary to perform the more general strategy, especially using input-dependent transmission. In such cases, restricted strategies may not be optimal. Furthermore, it is impossible to rule out such attacks by monitoring them in the protocol since the transmission can \textit{appear} to be input-independent when only conditioned on the inputs $(x,y)$, but can be input-dependent when also conditioned on the shared randomness $r$. Therefore, it is important to account for the input-dependence of the transmission $\eta$ carefully.

We first address the use of mixed states and general quantum channels.
It is not straightforward to perform the security analysis directly on CPTP maps and mixed states, so the idea would be to partially purify the final state $\rho_{RVA'B'}$ (note we do not apply the dimension bound on the combined classical and quantum output systems).
Let us first define the set of quantum states and channels we are examining.
The set $\vars{S}_q$ is defined as the set of (mixed) quantum states of dimension $2^q$, while $\vars{S}_q^p$ is defined as the set of pure quantum states with dimension $2^q$.
The set $\vars{C}_q$ is defined as the set of CPTP maps that maps $2^q$-dimension quantum states to $2^q$-dimension quantum states, while $\vars{C}_q^U$ is defined similarly for the set of unitaries.

Partial purification can be performed with purification of the mixed state $\rho_{AB}^r$ by doubling its number of qubits~\cite{Wilde2017}, while lifting the quantum channel to a higher dimensional unitary can be performed via the Stinespring dilation theorem~\cite{Wilde2017}. 
We can summarize the purification as a theorem.
\begin{theorem}[Purified strategy]
\label{thm:purification}
    Any state $\rho_{R A'B'V}^{xy}=\sum_rp_r\dyad{r}_R\otimes\channel_{AQ}^{xr}(\dyad{\Phi^+}_{VQ}\otimes\channel_B^{yr}(\rho_{AB}^r))$ with $\rho_{AB}^r\in\vars{S}_{2q}$, $\channel_{AQ}^{xr}\in\vars{C}_{q+1}$ and $\channel_B^{yr}\in\vars{C}_q$ can be purified with purification systems $P$ of dimension $2^{2q}$, $P_A$ of dimension $2^{2(q+1)}$ and $P_B$ of dimension $2^{2q}$, i.e. there exists a state $\ket{\psi^r}_{ABP}\in\vars{S}_{4q}^p$ and unitaries $U_{AQP_A}^{xr}\in\vars{C}_{3(q+1)}^U$ and $U_{BP_B}^{yr}\in\vars{C}_{3q}^U$ such that
    \begin{equation*}
    \begin{gathered}
        \rho_{RA'B'VPP_AP_B}=\sum_rp_r\dyad{r}_R\otimes[(U_{AQP_A}^{xr}\otimes U_{BP_B}^{yr})        (\dyad{\Phi^+}_{VQ}\otimes\dyad{\psi^r}_{ABP}\otimes\dyad{0}_{P_AP_B})(U_{AQP_A}^{xr\dagger}\otimes U_{BP_B}^{yr\dagger})]
    \end{gathered}
    \end{equation*}
    and $\Tr_{PP_AP_B}[\rho_{RA'B'VPP_AP_B}]=\rho_{RA'B'V}$. Furthermore, the maximum score of the original strategy is upper bounded by the maximum score of the purified strategy.
\end{theorem}
\begin{proof}
Refer to Methods.
\end{proof}

The fact that purification does not reduce the maximum achievable score allows us to reduce the analysis of the mixed state and quantum channel attacks to one with pure states and unitaries.
Since the number of qubits influences the size of the $\delta$-nets formed, we quantify the size of such nets when quantum channels and mixed states are utilized,
\begin{theorem}
    Consider the set of $q$-qubit general strategies (with $q$ qubits in subsystems $A$ and $B$ respectively). Then, there exists a $\delta$-covering net of the set of purified states in the Euclidean norm with covering number $\abs{\vars{N}_S}$ and $\delta$-covering nets for the set of unitaries corresponding to the quantum channels $\channel_{AQ}^{xr}$ and $\channel_{B}^{yr}$ in the operator norm with covering numbers $\abs{\vars{N}_A}$ and $\abs{\vars{N}_B}$ respectively, where
    \begin{gather*}
        \log_2\abs{\vars{N}_S}\leq 2^{4q+1}\log_2\left(1+\frac{2}{\delta}\right),\quad
        \log_2\abs{\vars{N}_A}\leq 2^{6q+7}\log_2\left(1+\frac{2}{\delta}\right),\quad
        \log_2\abs{\vars{N}_B}\leq 2^{6q+4}\log_2\left(1+\frac{2}{\delta}\right).
    \end{gather*}
\end{theorem}
\begin{proof}
Apply Theorem~\ref{thm:covering_number_state} for purified state of dimension $4q$ and Theorem~\ref{thm:covering_number_unitary} for unitary of dimension $3(q+1)$ for Alice and $3q$ for Bob.
\end{proof}
We note that in the proof of Theorem~\ref{thm:covering_number_unitary}, we have not made use of the fact that $U_A$ and $U_B$ are purifications of their corresponding CPTP maps.
This purification property implies that only the first $2^q$ rows of the unitary matrix are relevant to describe the strategy since the states are always initialized as $\ket{0}$ in systems $P_AP_B$.
As such, it may be possible to further reduce the covering number by restricting the set of unitaries discussed.
For simplicity, we leave any such optimization for future work.

The partitioning strategy here aims to address the assumption where loss $\eta$ is input-independent, and to account for shared randomness.
When analyzing $q$-qubit pure strategies, the set of high score strategies is defined with the idea of $(C,l)$-perfect strategies, where for $l$ pairs of strings $(x,y)$, the attacks declare no photon detection with probability $1-\eta$ and the score is upper bounded by $C$.
We note that this set of strategies may not encompass all optimal attacks.
One particular set of attacks not considered are attacks with loss that varies with $(x,y)$ pairs, and it is unclear if selecting a strategy with the same $\eta$ for all $(x,y)$ is the optimal strategy.

We also note that in many protocols, it is typical that shared randomness does not provide an advantage to adversaries, allowing us to simplify analysis to a single strategy without shared randomness.
In fact, this is the case for a lossless QPV~\cite{Bluhm2022} since $\sum_rp_r$ commutes with the maximization in Eq.~\eqref{eq:Winning_prob_opt}.
This allows us to argue that for any strategy, it would be optimal to pick one corresponding to an $r$ with the highest conditional score.
The same argument may no longer hold with loss since $\eta$ can vary with $r$, and the optimal score can vary non-linearly with $\eta$.
As such, an adversary mixing two strategies, one with higher loss and lower score and one with lower loss and higher score may have a overall larger score compared to a strategy with an average loss.

We begin our analysis by defining the set of strategies in a different manner, without assuming that $\eta$ is independent of $(x,y)$ and $r$.
\begin{definition}
    For any strategy, the input and shared randomness dependent score and transmission are defined as
    \begin{gather*}
        C_{rxy}=\sum_{s=0}^1 \sum_{z=0}^1 (-1)^{s+z}\Tr[(\Lambda_s^{f(x,y)}\otimes A_{z}^{xyr}\otimes B_{z}^{xyr})\rho_{V A'B}^{xyr}]\\
        p_{t|rxy}=1-\Tr[(\mathbb{I}_V\otimes A^{xyr}_{\perp}\otimes B_{\perp}^{xyr})\rho_{V A'B'}^{xyr}]
    \end{gather*}
    respectively, and the matching condition is defined as
    \begin{equation*}
        p_a=\Tr[(\mathbb{I}_V\otimes A^{xyr}_{z}\otimes B_{z'}^{xyr})\rho_{V A'B'}^{xyr}]=0,\,\forall z\neq z'.
    \end{equation*}
    A $q$-qubit general strategy is a $(\{p_r\}_r,C,\eta)$-strategy if it satisfies $2^{-2n}\sum_{xy}C_{xy}=C$, $2^{-2n}\sum_{xy}\eta_{xy}=\eta$, and the matching condition for all $(x,y)$. Furthermore, the strategy has a $(\cthnorm,\etath,\{l_{1,r},l_{2,r},l_{3,r},l_{4,r}\}_r)$-partition if for each $r$, there are:
    \begin{enumerate}
        \item $l_{1,r}$ pairs of $(x,y)$ satisfying $C_{rxy}\geq p_{t|rxy} \cthnorm$, $p_{t|rxy}\geq\etath$,
        \item $l_{2,r}$ pairs of $(x,y)$ satisfying $C_{rxy}\geq p_{t|rxy} \cthnorm$, $p_{t|rxy}<\etath$,
        \item $l_{3,r}$ pairs of $(x,y)$ satisfying $C_{rxy}<p_{t|rxy} \cthnorm$, $p_{t|rxy}\geq\etath$,
        \item $l_{4,r}$ pairs of $(x,y)$ satisfying $C_{rxy}<p_{t|rxy} \cthnorm$, $p_{t|rxy}<\etath$.
    \end{enumerate}
    We term the strategy for each $r$ as a sub-strategy with $(\cthnorm,\etath,l_{1,r},l_{2,r},l_{3,r},l_{4,r})$-partition.
\end{definition}
In this definition, there are no explicit assumptions on the $r$ and $(x,y)$-dependency of loss and error, and all attack strategies can be described by this set of strategies.
Note that we are free to choose the threshold $\cthnorm$ and $\etath$.

Following the idea from the original security analysis~\cite{ABB23,Llorenc2023}, we seek to provide an upper bound on the number of high conditional rounds, specifically $l_{1,r}$ for any strategy.
As such, let us focus on this partition and define the set of states that can satisfy the corresponding conditions for different basis values $f(x,y)$.
\begin{definition}\label{def: normalised C-eta set}
    Let $\tilde{\eta}\in[0,1]$ and $\hat{C}'\in[-1,1]$.
    The set of output states for a partition for a basis is defined as
    \begin{gather*}
        \hat{\vars{S}}_{\alpha'}^{\hat{C}',\tilde{\eta}}=\{\ket{\Psi}_{V A'B'}=(\mathbb{I}_V \otimes U_{QA}\otimes U_B)(\ket{\Phi^+}_{VQ}\otimes \ket{\psi}_{AB}):\exists \{A_z\}_z,\{B_z\}_z,\, s.t.,\,C_{\alpha'}\geq p_t \hat{C}',\,p_t\geq \tilde{\eta},p_a=0\},
    \end{gather*}
    where $U_{QA}$, $U_B$ and $\ket{\psi}_{AB}$ are restricted in the same manner as described in Theorem~\ref{thm:purification}, i.e. $U_{QA}$ and $U_B$ are purifications of CPTP maps for $q+1$ and $q$ qubits respectively, and $\ket{\psi}_{AB}$ remains a purification of a mixed $2q$-qubit state.
\end{definition}
For the partition of interest, these sets can be defined with $\hat{C}'=\cthnorm$ and $\tilde{\eta}=\etath$.
Using the same arguments for a tight trace distance bound made when analyzing $q$-qubit pure strategy, we demonstrate in Cor.~\ref{cor: existence trace distance bound for normalised score} that any state drawn from one set will be far from at least one other set.

The third step of the security proof involves the formation of a classical rounding.
\begin{definition}[Classical Rounding, modified from Ref.~\cite{Bluhm2022,Llorenc2023}]
    A function $g:\{0,1\}^{3k}\rightarrow\{\vars{S}:\vars{S}\subseteq[m],\abs{\vars{S}}=m-1\}$ is termed a $(\etath,\cthnorm)$-classical rounding of size $k$ if for any function choice $f:\{0,1\}^{2n}\rightarrow[m]$, any $l_{1,r}\in[2^{2n}]$, any sub-strategy with $(\cthnorm,\etath,l_{1,r},l_{2,r},l_{3,r},l_{4,r})$-partition, there exist functions $f_A:\{0,1\}^n\rightarrow\{0,1\}^k$, $f_B:\{0,1\}^n\rightarrow\{0,1\}^k$ and $\lambda\in\{0,1\}^k$ such that $f(x,y)\in g(f_A(x),f_B(y),\lambda)$ for $l_{1,r}$ pairs of $(x,y)$.
\end{definition}
We note that in this case we focus on each sub-strategy instead of the full strategy, but each sub-strategy remains a valid strategy.
This allows us to construct a classical rounding.
\begin{theorem}
\label{thm:classical_rounding_exist_gen}
    Suppose $\cthnorm$ and $\etath$ are such that $\tilde{\delta}(\cthnorm,\etath)>0$ and let $\delta<\frac{\tilde{\delta}}{6}$. Then, there exists a classical rounding of size $k=2^{6q+7}\log_2\Big\lceil2+\frac{12}{\tilde{\delta}}\Big\rceil$.
\end{theorem}
\begin{proof}
Same as proof of Theorem~\ref{thm:rounding_from_delta_sec}, except that we use sets $\vars{S}_{\alpha'}^{\cthnorm,\etath}$ and $l_{1,r}$ for each sub-strategy.
The size $k$ follows from the modifications from purification of the $q$-qubit mixed strategy.
\end{proof}

We follow the original security proof and demonstrate that when a random function is used and when $q$ is bounded, there is a bound on the number of input pairs in the high conditional score and high transmission partition, $l_{1,r}$.
\begin{theorem}
\label{thm:nu_value_compute_gen}
    Fix a $(\etath,\cthnorm)$-classical rounding with $k=2^{6q+7}\log_2\Big\lceil2+\frac{12}{\tilde{\delta}}\Big\rceil$ and let $q\leq\frac{1}{6}n-q_0$. Then, a uniform random function $f:\{0,1\}^{2n}\rightarrow[m]$ fulfills the following with probability at least $1-2^{-\beta}$: For any $f_A$, $f_B$ and $\lambda$, $f(x,y)\in g(f_A(x),f_B(y),\lambda)$ holds for less than $2^{2n}(1-\nu)$ pairs of $(x,y)$, for 
    \begin{equation*}
        \nu= h_b^{-1}\left\{\log_2\left(\frac{m}{m-1}\right)-2^{9-6q_0}\log_2\Bigg\lceil2+\frac{12}{\tilde{\delta}}\Bigg\rceil-\frac{\beta}{2^{2n}}\right\},
    \end{equation*}
    where $h_b(x)$ is the binary entropy function.
\end{theorem}
\begin{proof}
Same proof as Theorem~\ref{thm:nu_value_compute}, except using a different $k$.
\end{proof}
Consequently, for any sub-strategies with $l_{1,r}>2^{2n}(1-\nu)$, they require a large number of qubits to successfully implement.
\begin{theorem}
\label{thm:sub_strategy_req}
    A uniform function $f:\{0,1\}^{2n}\rightarrow [m]$ has the following with probability at least $1-2^{-\beta}$: Any sub-strategy with $l_{1,r}>2^{2n}(1-\nu)$ requires $q>\frac{1}{6}n-q_0$ to implement.
\end{theorem}
\begin{proof}
From Theorem~\ref{thm:classical_rounding_exist_gen}, we know that a classical rounding exists for suitably chosen $\etath$ and $\cthnorm$.
Assume there exists a sub-strategy with $l_{1,r}>2^{2n}(1-\nu)$ that requires $q\leq\frac{n}{6}-q_0$ to implement.
By definition, it implies that the sub-strategy satisfies $C_{rxy}\geq p_{t|rxy} \cthnorm$ and $p_{t|rxy}\geq \etath$.
From the classical rounding, it implies that we can find a $f_A$, $f_B$ and $\lambda$ such that $f(x,y)\in g(f_A(x),f_B(y),\lambda)$ for at least $2^{2n}(1-\nu)$ pairs of $(x,y)$.

Since the sub-strategy requires only $q\leq\frac{n}{6}-q_0$ to implement, by Theorem~\ref{thm:nu_value_compute_gen}, we necessarily have that for a randomly selected function $f$, $f(x,y)\in g(f_A(x),f_B(y),\lambda)$ for less than $2^{2n}(1-\nu)$ pairs of $(x,y)$ with probability of at least $1-2^{-\beta}$.
This indicates that except with probability less than $2^{-\beta}$, the two statements contradict.
As such, for a random function selection, the assumption is not true, i.e. sub-strategies with $l_{1,r}>2^{2n}(1-\nu)$ require $q>\frac{n}{6}-q_0$ to implement, with probability at least $1-2^{-\beta}$.
\end{proof}

We immediately have a corollary that if we restrict the sub-strategies to have $q\leq\frac{1}{6}n-q_0$, then the sub-strategies have $l_{1,r}\leq 2^{2n}(1-\nu)$ with high probability.
\begin{corollary}
\label{cor:sub_strategy_lim}
    Let the choice of $f$ be a random function, and Alice and Bob are restricted to strategies with $q\leq\frac{1}{6}n-q_0$. Then, for any sub-strategy, $l_{1,r}\leq 2^{2n}(1-\nu)$ except with probability of $1-2^{-\beta}$.
\end{corollary}
This allows us to upper bound the score for any sub-strategy.

We first consider the maximum score for a sub-strategy with $l_{1,r}\leq 2^{2n}(1-\nu)$, noting that the score is upper bounded by the transmission,
\begin{theorem}
\label{thm:ABE Bound}
    For any general adversary strategy with transmission $\eta_{rxy}$, $C_{rxy}\leq\eta_{rxy}$.
\end{theorem}
\begin{proof}
    Refer to Methods.
\end{proof}
When the transmission of the sub-strategy, $\eta_r$, is high, there is a limit to the number of low transmission $(x,y)$ pairs, i.e. $l_{2,r}$ is upper bounded.
Since both $l_{1,r}$ and $l_{2,r}$, the number of $(x,y)$ pairs with high scores are bounded, we can compute an upper bound to the achievable score,
\begin{theorem}
\label{thm:sub_strategy_err}
    Let the choice of $f$ be a random function, and let Alice and Bob be restricted to strategies with $q\leq\frac{1}{6}n-q_0$. For any sub-strategy with $(\cthnorm,\etath,l_{1,r},l_{2,r},l_{3,r},l_{4,r})$-partition, $l_{1,r}\leq 2^{2n}(1-\nu)$ and average transmission $\eta_r$, the score $C_r$ is upper bounded by $\min\left\{\eta_r-\etath(1-\cthnorm)\left[\frac{\eta_r-\etath}{1-\etath}-1+\nu\right],\eta_r\right\}$.
\end{theorem}
\begin{proof}
Refer to Methods.
\end{proof}

Since the upper bound for each sub-strategy is linear in $\eta_r$, we can combine them to compute an upper bound on the overall score for any strategy.
\begin{theorem}
\label{thm:Error_Rate_QPV_Entanglement}
    Let $f$ be a random function, and let Alice and Bob be restricted to $q$-qubit general strategies with $q\leq\frac{1}{6}n-q_0$.
    For any $(\{p_r\}_r,C,\eta)$-strategy with $(\cthnorm,\etath,\{l_{1,r},l_{2,r},l_{3,r},l_{4,r}\}_r)$-partition, with probability at least $1-2^{-\beta}$, the score is upper bounded by $C\leq\eta-\etath(1-\cthnorm)\left(\frac{\eta-\etath}{1-\etath}-1+\nu\right)$.
\end{theorem}
\begin{proof}
From Corollary~\ref{cor:sub_strategy_lim}, every sub-strategy involved in this overall strategy has $l_{1,r}\leq 2^{2n}(1-\nu)$ except with probability $2^{-\beta}$.
As such, from Theorem~\ref{thm:sub_strategy_err}, the score for each sub-strategy is given by
\begin{equation}
    C_r\leq\min\left\{\eta_r-\etath(1-\cthnorm)\left[\frac{\eta-\etath}{1-\etath}-1+\nu\right],\eta_r\right\}.
\end{equation}
We can simply bound the overall score as
\begin{equation}
\begin{split}
    C=&\sum_rp_rC_r\\
    \leq&\sum_rp_r\min\left\{\eta_r-\etath(1-\cthnorm)\left[\frac{\eta_r-\etath}{1-\etath}-1+\nu\right],\eta_r\right\}\\
    \leq&\sum_rp_r\eta_r-\etath(1-\cthnorm)\left[\frac{\sum_rp_r\eta_r-\etath}{1-\etath}-1+\nu\right]\\
    =&\eta-\etath(1-\cthnorm)\left[\frac{\eta-\etath}{1-\etath}-1+\nu\right],
\end{split}
\end{equation}
where the linearity of the upper bound on the score allows $r$ to be removed.
\end{proof}
Since the strategy described in the theorem includes all possible $q$-qubit general strategies, we obtain a valid upper bound on the achievable score for any such adversary for each $\eta$.
Therefore, if a larger score is observed in the experiment in the asymptotic regime, we can be confident that there must be a party present at $P$ participating in the QPV protocol.

Since both purification and transmission and score partitioning modify the security analysis to tackle a stronger adversary, they necessarily reduce the protocol's performance.
Purification directly impacts the the classical rounding size $k$, and this results in security against an adversary with lower number of qubits, $q=\frac{1}{6}n-q_0$ compared to $q=\frac{1}{2}n-q_0$, and a change in the second term of $\nu$.
We note that this term, with coefficient $2^{9-6q_0}$ goes to $0$ in the large $n$ regime ($q_0\rightarrow\infty$).

The score and transmission partitioning significantly degrade both loss tolerance and error rate.
For non-zero error rate tolerance, the score upper bound must satisfy $C^{UB}\leq\eta$, and from Theorem~\ref{thm:Error_Rate_QPV_Entanglement}, this requires that $\frac{\eta-\etath}{1-\etath}\geq 1-\nu$.
Since $\etath\geq\frac{1}{m}$ for there to be non-zero overlap in the sets $\hat{\vars{S}}_{\alpha'}^{\cthnorm,\etath}$, the loss tolerance can be bounded by
$\eta\geq1-\frac{m-1}{m}\nu$, where $\nu=h_b^{-1}(\log_2(\frac{m}{m-1}))$.
For $m=2$, this limits the loss tolerance to at most \SI{75}{\percent}, and for $m=3$, the loss tolerance is at \SI{90.6}{\percent}.
This deterioration of loss tolerance would make practical QPV implementation challenging, highlighting the need for improved security analysis techniques, which we leave for future work.

\begin{figure}[!h]
    \centering
    \includegraphics[width=0.8\linewidth]{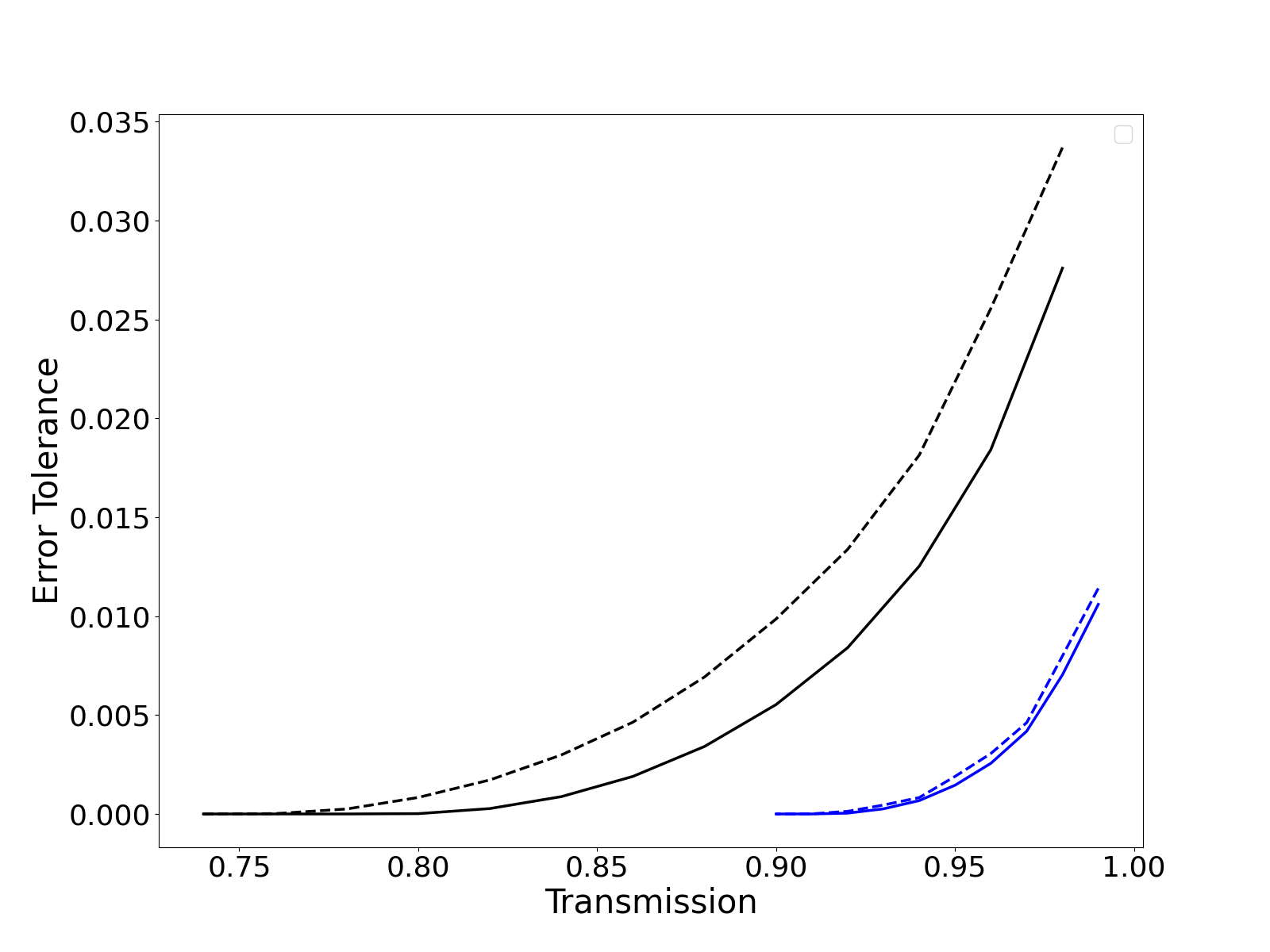}
    \caption{Plot of error tolerance at various transmission for secure position verification against a generalized entangled adversary. The simulation is performed for 2 basis (black line) and 3 basis (blue lines), and for $q_0=3 $ (solid lines) and $q_0=100$ (dotted lines).}
    \label{fig:Gen_Entangled_Adv_Plot}
\end{figure}

To illustrate the impact, we numerically optimize the score upper bound $C^{UB}$ described in Theorem~\ref{thm:Error_Rate_QPV_Entanglement} over $\cthnorm$ and $\etath$. 
For simplicity, we fix $\frac{\beta}{2^{2n}}=10^{-10}$ and map the maximum score to the minimum error rate.
Figure~6 shows severe degradation of performance even at large $q_0$, with loss tolerance being at $\eta>0.75$ for 2 bases and $\eta>0.9$ for 3 bases.

\textbf{Current Guarantees and Remaining Gaps}:
The security results can be summarized as a corollary,
\begin{corollary}
Let $f$ be a random function.
Consider the multi-basis protocol where the verifiers certify that a party is at location $P$ if the score exceeds $C>\eta-\etath(1-\cthnorm)\left[\frac{\eta-\etath}{1-\etath}-1+\nu\right]$, where $\nu$ is as described in Theorem~\ref{thm:nu_value_compute_gen}.
Any adversary who:
\begin{enumerate}
    \item May share arbitrary (possibly unbounded) classical randomness.
    \item May prepare pre-shared mixed quantum states $\rho_{AB}$ of bounded dimension $2^{2q}$.
    \item May apply general quantum channel $\channel_{AQ}^{xr}$ with input and output dimension of $2^{q+1}$ and channel $\channel_B^{yr}$ with input and output dimension of $2^q$.
    \item May implement transmission strategies that depend on both the inputs $(x,y)$ and the shared randomness $r$.
    \item Is bounded in dimension by $q\leq\frac{1}{6}n-q_0$,
\end{enumerate}
can only be falsely certified to be at location $P$ with probability less than $2^{-\beta}$.
\end{corollary}

While the current analysis provides several generalizations, there remains gaps to address for future work:
\begin{enumerate}
    \item The framework does not account for adversaries with unbounded classical output size (but bounded quantum output size) from quantum channels $\channel_{AQ}^{xr}$ and $\channel_B^{yr}$.
    \item The framework does not incorporate the tightened security analysis for QPV presented in Ref.~\cite{Llorenc2025}. How to extend their approach to lossy channels or incorporate our generalizations into their security analysis remain open questions.
\end{enumerate}

\subsection{Application: Exchanging Keys with Location Credentials}

Beyond position verification, QPV can serve as a building block for position-based cryptography. For instance, in Ref.~\cite{Buhrman2014}, the authors developed a message authentication protocol from QPV, and demonstrated its application to position-based key exchange.
The key exchange component relies on QKD~\cite{Portmann2022_Security,Xu2020_ExptQKDReview}, which is a secure means of key exchange between two parties, Alice and Bob.
(We note that Alice and Bob here refer to the parties intending to exchange keys and not the adversaries in QPV. We keep this naming convention for simplicity since this is the convention of QKD and QPV, and it should be clear who Alice and Bob represent from the context.)
QPV-based message authentication then authenticates the messages exchanged between Alice and Bob, using their location as credentials.
While position-based key exchange is valuable if location information is the only credential to be used, the complexity of QPV implementation and its weaker security guarantee compared to the commonly used message authentication protocol via Wegman-Carter~\cite{Carter1977_WegmanCarter} means it is unlikely to serve as an alternative authentication method in QKD.

However, there remain scenarios where QPV-based authentication can be helpful.
Currently, most QKD relies on Wegman-Carter~\cite{Carter1977_WegmanCarter}, which requires Alice and Bob to have pre-shared keys that would be used to hash the message to generate a tag to check for message tampering.
After each QKD round, part of the shared keys can be recycled~\cite{Portmann2014} while the remainder is obtained from part of the QKD keys generated from that round.
This leads to two issues: (1) Alice and Bob need to pre-share a key for the first QKD round and (2) using QKD keys for message authentication causes the soundness parameter to accumulate.
The latter issue stems from the composability framework of QKD, where QKD keys are only $\varepsilon$-secure, and using these keys in the next round where perfectly secure authentication keys are required leads to $2\varepsilon$-security in the next round -- thus accumulating. 
Furthermore, authentication keys can also be consumed during denial-of-service attacks, which necessitates alternative means of authentication.
To address these issues, one has to rely on other methods of authentication for a fresh QKD round, which could be via (1) manual transfer of master keys between Alice and Bob by a trusted party for use in authentication, (2) public key infrastructure, and (3) QPV-based message authentication.
After this fresh QKD round, the key exchange process can revert to using Wegman-Carter and QKD keys until the next refresh.

All three alternative methods of authentication have their limitations, and a list of some aspects to consider is shown in Table~2.
Manual key transfer requires full trust in the party carrying the key, while the use of PKI only guarantees computational security of the message authentication scheme and QPV-based authentication is complex to implement.
Therefore, the choice of authentication method may depend on the application, and the use of a complex but more secure QPV-based authentication remains a viable option for this infrequent but important process.

Here, we formalize the security of key exchange with location credentials discussed in Ref.~\cite{Buhrman2014}, and make improvements to the proposal to reduce the QPV runs (1 QPV run is a complete QPV protocol, which may involve multiple rounds) required, which should help reduce the complexity of the implementation.
The upgrades come from two aspects: (1) modification of the QPV-based message authentication protocol and (2) modification of QKD to reduce message authentication requirements.

\begin{table*}[h]
    \centering
    \includegraphics[width=\linewidth]{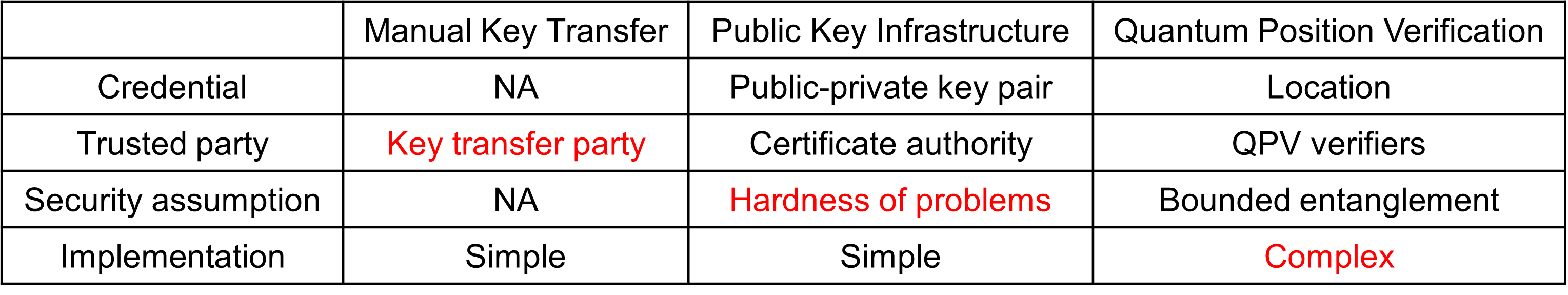}
    \caption{Comparison of aspects of the alternative authentication options. Red text highlights the least desirable option for each aspect.}
    \label{tab:QPKE_Comparison}
\end{table*}

The QPV-based message authentication protocol in Ref.~\cite{Buhrman2014} assigns the pass/fail of QPV runs as bits 1 and 0, and demonstrates security when proper encoding is utilized.
Intuitively, security follows from the fact that an adversary not at position $P$ cannot pass QPV and is thus unable to switch bits of the codeword from bit 0 to 1.
Here, we present an improvement to the protocol by combining it with message authentication codes -- using QPV to authenticate the key instead of the full message.
This involves the prover sending the message $M$ and a tag $T=h_K(M)$ using a $\delta$-almost 2-universal hash family $\{h_k:\{0,1\}^n\rightarrow\{0,1\}^{l_T}\}_{k\in\vars{K}}$ to the verifier.
The QPV-based message authentication is used to send $K$ in an authenticated manner, which in turn is used by the verifier to verify the tag $T$.
Since there exists 2-universal hash families with small key length, $l_K=2\lfloor l_T+\log_2(\frac{n}{l_T})+1\rfloor$ with $\delta=2^{-l_T+1}$~\cite{Tomamichel2011_QLHL}, this can reduce the QPV runs required from $O(n)$ to $O(\log_2n)$, for messages of length $n$.
Details of the proposed QPV-based message authentication and the proof of its security can be found in Section I of the supplementary information.

Standard QKD protocols typically require authentication of every message transmission step, which requires a significant number of QPV runs.
Here, we reduce this overhead by first selecting QKD protocols where message authentication is left to the final two communication steps~\cite{Kiktenko2020_Lightweight,QAKE} -- one from Alice to Bob and one from Bob to Alice.
Furthermore, the latter communication step only involves a single bit response from Bob to Alice, which requires a single QPV run  to be authenticated.
This reduces the required QPV runs to only $O(log_2n)+1$.

Details of the proposed key exchange with location credentials protocol can be found in Section II of the supplementary information, along with the formal security analysis, where we are able to demonstrate that the proposed protocol is secure.

Beyond supporting QKD, QPV enables other forms of position-based key exchange, where the unique nature of using a location credential provides anonymity to parties.
For example, a server providing data access only to clients that are at a certain location (e.g. an office building) can preserve client anonymity.
In this scenario, the client and server can exchange keys, with the server using the client's location as credential and the client using other means (e.g. PKI) to verify the server's identity in the key exchange process.
In Section III of the supplementary information, we present and formally prove the security of such a protocol which requires only a single QPV run.

\section{Discussion}

Interestingly, the ABE score was chosen in Ref.~\cite{sekatski2025certification} as it can be used to construct a steering inequality~\cite{Uola2020}. In quantum steering, the verifier's sub-normalized states $\left\{\tilde\rho_{z}^{\alpha'} = \Tr_Q \left[\ketbra{\Phi^+}{\Phi^+}_{VQ} (\mathbb{I}_V \otimes P_{z}^{\alpha'}) \right] \right\}_{\alpha', z}$ are called the \textit{assemblage} created by the prover's measurement $\{P_z^{\alpha'}\}_{\alpha', z}$. If the assemblage can be written as
\begin{equation}
    \tilde\rho_{z}^{\alpha'} = \sum_{\lambda} p(z|\alpha', \lambda)\ p_{\lambda} \rho_{\lambda}, \qquad \forall \alpha', z
\end{equation}
for some conditional probability distributions $\{p(z|\alpha',\lambda)\}_{z, \alpha', \lambda}$, $\{p_{\lambda}\}_{\lambda}$ and states $\{\rho_{\lambda}\}_{\lambda}$, we say that the assemblage admits a \textit{local hidden state} model. Otherwise, we say that the assemblage is \textit{steerable} by the prover. Now, suppose that the statistics characterized by $\{p(s,z|\alpha, \alpha')\}_{s,z,\alpha,\alpha'}$ is compatible with the local hidden state model:
\begin{equation}
    p(s,z|\alpha, \alpha') = \Tr[ \tilde\rho_{z}^{\alpha'} \Pi_{s}^{\alpha}] = \sum_{\lambda} p_{\lambda} p(z|\alpha',\lambda) \Tr[\rho_{\lambda} \Pi_{s}^{\alpha}],
\end{equation}
then there exists some POVM $\{M_{\lambda}\}_{\lambda}$ such that
\begin{equation}
    p_{\lambda} \rho_{\lambda} = \Tr_{Q}\left[\ketbra{\Phi^+}{\Phi^+}_{VQ} (\mathbb{I}_V \otimes M_{\lambda})\right], \qquad \forall \lambda.
\end{equation}
 
In a QPV protocol, Alice can use the same POVM $\{M_{\lambda}\}_{\lambda}$ and measure the register $Q$ before receiving the inputs $x$ and $y$. Suppose that Alice obtains the outcome $\lambda$, for each $\alpha' \in [m]$, she can sample a guess $\hat{z}_{\alpha'}$ from the distribution $\{p(z|\alpha',\lambda)\}_{z}$. She can then forward the guesses $\{\hat{z}_{\alpha'}\}_{\alpha'}$ to Bob at the speed of light. Therefore, by the time Alice and Bob receive both inputs, $x$ and $y$, both parties can reply to the respective verifier with $\hat{z}_{f(x,y)}$. Clearly, this will reproduce the original distribution $\{p(s,z|\alpha, \alpha')\}_{s,z,\alpha,\alpha'}$, and hence we conclude that whenever the statistics admit a local hidden state model, then there exists an attack on QPV that reproduces those statistics. This establishes a link between QPV and quantum steering, and it would be interesting to investigate the converse: when an assemblage is steerable by the prover (i.e., the statistics do not admit the local hidden state model), would that rule out any attacks on QPV (perhaps, without pre-shared entanglement) that reproduce the same statistics? An affirmative answer would allow a novel technique to analyze the security of QPV protocols. We leave this investigation for future work.

We also briefly analyze the feasibility of implementing multi-basis QPV with current technology.
For the quantum components, the numerical results against entangled adversaries (Figure~5) show the required error rate and transmission for secure implementation.
Additionally, the prover's measurement in basis $\alpha'$ must be performed with minimal delay.
Such stringent requirements on loss along with high-speed measurement are similar to those in device-independent cryptography.
For instance, in Ref.~\cite{Liu2021}, a high-fidelity (\SI{99.1}{\percent}) entangled source with high heralding efficiency (\SI{82.2}{\percent}) was prepared and one subsystem was transmitted to another party at some distance with overall detection efficiency (ignoring fiber loss, but including heralding efficiency) at \SI{68.3}{\percent}.
The detection is performed with a high-speed Pockels cell for basis selection, with overall measurement delay of \SI{230}{\nano\second}, which corresponds to a positional uncertainty of \SI{69}{\metre}.
Assuming that the same entangled source can be modified to generate a Bell state with the same fidelity, the verifier can use the Bell state to prepare $\ket{\psi^{\alpha}_s}$ by measuring the first photon in basis $\alpha$ (Note that the entanglement source in Ref.~\cite{Liu2021} generates the entangled state $\cos(24.3^\circ)\ket{HV}+\sin(24.3^\circ)\ket{VH}$. Based on the design of the source, the coefficients of the $HV$ and $VH$ components can be tuned by adjusting the waveplates orientations to control the polarization of the light entering the Sagnac loop, so it is reasonable to assume that the fidelity can be maintained for a class of entangled state preparation, which includes the Bell state at angle \SI{45}{\degree}).
This leads to an error rate of \SI{0.6}{\percent} when both prover and verifier measure in the same basis (assuming a Werner state preparation).
For transmission through an optical fiber with loss $\eta_{ch}\sim10^{-0.2L/10}$ over a distance $L$ (in km), Figure~5 shows that secure implementation against entangled adversaries (but not generalized) is possible up to $L=$\SI{6.8}{\kilo\metre} (total effective loss of about \SI{50}{\percent} for $q_0=5.5$).
This improves upon earlier analysis with modified SDP, which achieves $L=$\SI{2.8}{\kilo\metre} (total effective loss of \SI{60}{\percent}).
The distance may be increased further by increasing the number of basis measurements on the prover side.

We note that the performance improvements discussed above are based on the improved security analysis against $q$-qubit restricted strategies, which does not account for the generalized adversary strategies. Against $q$-qubit general strategies, the proposed experimental setup would not achieve security as the transmission falls below the 75\% threshold required for security (see Figure~6). This contrast highlights a key tension: while our improved security analysis can significantly improve error and loss tolerance, these gains are realized only under more restricted adversary assumptions. When the adversary model is broadened to include features such as input-dependent loss and unbounded shared randomness, the security requirements become much more demanding and the achievable performance degrades substantially.

The implementation also imposes requirements on the classical system. 
The verifier is required to send the classical messages to the prover with minimal delay -- which could in principle be close to the speed of light via free-space communication.
In addition, the prover is required to receive the messages and process the received messages quickly.
The speed of the computation depends on the function $f$. For simple functions like inner product, computation requires only a few clock cycles. Using a field programmable gate array (FPGA) with clock cycles of \SI{500}{\mega\hertz} yields a delay of $<$\SI{20}{\nano\second} (location uncertainty of $<$\SI{6}{\metre}).
Other classical delays exist but are generally small relative to the distances (\SI{6.8}{\kilo\metre}) that QPV can be implemented over.

While the brief analysis above suggests that QPV is feasible, there can be more considerations in practice.
For instance, the protocol proposed is difficult to scale to longer distances since the number of bases has to be increased to improve loss tolerance.
One would also have to consider finite-size effects, which may further exacerbate the issue, where a minimum number of rounds are required to distinguish the honest implementation if the error rate of the implementation and that for secure QPV is close.
Therefore, for practical implementations, techniques such as that in Ref.~\cite{ABB23}, where the prover is required to commit to whether a photon is received, may need to be adopted to address the increased channel loss.
We leave these investigations for future work.

In conclusion, we have presented several improvements to both QPV security and protocol design.
Firstly, we introduced a multi-basis QPV protocol in which the state preparation is restricted to six states while the prover measures in multiple bases.
This reduces the complexity of state preparation and improves the flexibility of protocol implementation, where performance can be easily adapted for different provers without modifying the verifier's device.
Decoupling state preparation from measurement basis (determined by $f(x,y)$) also enables third-party state preparation.
Secondly, we refined the security analysis against entangled adversaries through two improvements: (1) tightened trace distance bounds and (2) a modified classical rounding construction.
These improvements yield tighter security bounds, which can relax experimental requirements.
Thirdly, we generalized the adversary model to encompass mixed states, unbounded shared randomness and input-dependent transmission.
Notably, accounting for input-dependent transmission and unbounded shared randomness significantly degrades QPV performance, particularly loss tolerance.
This highlights the need for improved security analysis techniques, including addressing the use of general quantum channels with unbounded classical outputs.
Finally, we illustrate the potential utility of QPV beyond location verification, for example as an authentication mechanism in QKD.

\section{Methods}

\subsection{Preliminaries: Set Covering}

It is known that the upper bound of the covering number of a hypersphere with norm 1 can be given by~\cite{Vershynin2012}
\begin{theorem}
\label{thm:Norm_ball_covering}
    Let $\norm{.}$ be any norm on points $x\in\mathbb{R}^d$. The covering number for a $\delta$-net for a norm-ball of unit radius can be bounded by
    \begin{equation*}
        \abs{\vars{N}}\leq \left(1+\frac{2}{\delta}\right)^d
    \end{equation*}
\end{theorem}
As such, we can compute the covering number for the set of pure states,
\begin{theorem}
\label{thm:covering_number_state}
    The set of pure states in a Hilbert space with dimension $d$, i.e. $\{\ket{\psi}:\ket{\psi}\in\Hil_d\}$, has a $\delta$-covering net in the Euclidean norm with covering number
    \begin{equation*}
        \abs{\vars{N}_S}\leq \left(1+\frac{2}{\delta}\right)^{2d}
    \end{equation*}
\end{theorem}
\begin{proof}
The set of pure states in a Hilbert space with dimension $d$, can be described by
\begin{equation}
    \ket{\psi}=\sum_{j=1}^d(a_j+b_ji)\ket{j},
\end{equation}
for any orthonormal basis $\{\ket{j}\}_j$, with the normalization constraint $\sqrt{\sum_j(a_j^2+b_j^2)}=1$.
As such, the set of pure states forms a unit sphere in Euclidean norm, with $x=(\{a_j\}_j,\{b_j\}_j)$, i.e. $x\in\mathbb{R}^{2d}$, and with $\norm{x}_2=1$ for all states.
Therefore, by Theorem~\ref{thm:Norm_ball_covering}, $\abs{\vars{N}_S}\leq \left(1+\frac{2}{\delta}\right)^{2d}$.
\end{proof}

We can also compute the covering number for a unitary matrix in 2-norm, using ideas from Ref.~\cite{Caro2022}.
We note that 2-norm (or spectral norm) here is defined to be the operator norm induced from the 2-norm, i.e. $\norm{A}_{op}=\max_{\norm{x}_2\leq 1}\norm{Ax}_2$.
\begin{theorem}
\label{thm:covering_number_unitary}
The set of unitary matrices that act on a Hilbert space with dimension $d$, i.e. $\{U|U:\Hil_d\rightarrow\Hil_d\}$, has a $\delta$-covering net in the operator norm with covering number
\begin{equation*}
    \abs{\vars{N}_U}\leq\left(1+\frac{2}{\delta}\right)^{2d^2}.
\end{equation*}
\end{theorem}
\begin{proof}
Since all unitary operators do not alter the Euclidean norm of any vector, all unitary operators have operator norm 1, $\norm{U}_{op}=1$.
As such, we can describe the set of unitary operators $U=\sum_{jk}(a_{jk}+b_{jk}i)\dyad{j}{k}$ by a vector $x=(\{a_{jk}\}_{jk},\{b_{jk}\}_{jk})$, with $x\in\mathbb{R}^{2d^2}$, since the unitary contains $d^2$ complex entries.
The set of unitaries is constrained by $\norm{x}_{opvec}=1$, where $\norm{.}_{opvec}$ is a norm on the vector $x\in\mathbb{R}^{2d^2}$, and is computed by reforming $x$ into the corresponding unitary matrix and computing the corresponding operator norm.
Therefore, by Theorem~\ref{thm:Norm_ball_covering}, $\abs{\vars{N}_U}\leq\left(1+\frac{2}{\delta}\right)^{2d^2}$.
\end{proof}

\subsection{Proof of Theorem~\ref{thm:rounding_from_delta_sec}}

Let $\delta<\frac{\tilde{\delta}}{6}$.
Consider $\delta$-nets $\vars{N}_S$, $\vars{N}_A$ and $\vars{N}_B$, corresponding to that for the set of pure states of dimension $2q$ in Euclidean norm, the set of unitaries acting on a Hilbert space with dimension $q+1$ in operator norm, and the set of unitaries acting on a Hilbert space with dimension $q$ in operator norm.
Let $\ket{\phi_{\lambda}}\in\vars{N}_S$, $U^{x'}_{AQ}\in\vars{N}_A$ and $U^{y'}_B\in\vars{N}_B$, where $\lambda$, $x'$ and $y'$ label the choice of the nets.
Define a set that extends $\vars{S}^{\cth,\eta}_{i}$ by a $3\delta$-ball, $\vars{S}^{\cth,\eta}_{i,3\delta}=\{\ket{\psi}:\exists\ket{\phi}\in\vars{S}^{\cth,\eta}_{i},\Delta(\ket{\psi},\ket{\phi})\leq3\delta\}$.
We define a function $g$ that computes an output for each net with final state $\ket{\psi^{x'y'}}=(U_{AQ}^{x'}\otimes U_B^{y'})(\ket{\phi_{\lambda}}\otimes\ket{\Phi^+})$:
\begin{enumerate}
    \item If $\ket{\psi^{xy}}\notin\vars{S}^{\cth,\eta}_{i,3\delta}$ for all $i\in[m]$, choose set $g(x',y',\lambda)=[1,m-1]$.
    \item Otherwise, pick any $i\in[m]$ where $\ket{\psi^{x'y'}}\in\vars{S}^{\cth,\eta}_{i,3\delta}$ and find a state $\ket{\phi'}\in\vars{S}^{\cth,\eta}_i$ such that $\Delta(\ket{\phi'},\ket{\psi^{x'y'}})\leq 3\delta$, which by definition exists. By Cor.~\ref{cor: existence trace distance bound}, there exists $j\neq i$ such that $\forall\ket{\psi_j}\in\vars{S}^{\cth,\eta}_j$, $
    \Delta(\ket{\phi'},\ket{\psi_j})\geq \tilde{\delta}$. We can then set $g(x',y',\lambda)=[m]\setminus\{j\}$.
\end{enumerate}
By design, the set chosen is of size $m-1$ as required.
We note here that the set choice in step 1 does not matter since these do not form part of the $l$ $(x,y)$-pairs with high score.
We also note that since $\Delta(\ket{\phi'},\ket{\psi^{x'y'}})\leq 3\delta$ and $\Delta(\ket{\phi'},\ket{\psi_j})\geq \tilde{\delta}$, then it must be true that $\Delta(\ket{\psi^{x'y'}},\ket{\psi_j})\geq \tilde{\delta}-3\delta$.

What remains is to prove that the $g$ defined forms a $(\eta,\cth)$-classical rounding.
For any $(C,\eta)$-strategy with $(\cth,l)$-partition, we can define $x'=f_A(x)$, $y'=f_B(y)$, $\lambda$ as the labels indicating the closest unitary and state (of the nets) to $U_{AQ}^{x}$, $U_{B}^{y}$ and $\ket{\psi}_{AB}$.
For each of the $l$ $(x,y)$ pairs, the strategy satisfies $C_{xy}\geq \cth$, $p_{t|xy}=\eta$ and the matching condition.
As such, the state generated before measurement is given by $(U_{AQ}^{x}\otimes U_{B}^{y})\ket{\psi'}\in\vars{S}^{\cth,\eta}_{f(x,y)}$, where we define $\ket{\psi'}=\ket{\psi}\otimes\ket{\Phi^+}$.
By definition of the $\delta$-net, we can always find the closest unitary and state $U^{x'}_{AQ}$, $U^{y'}_B$ and $\ket{\phi_{\lambda}}$, each of which is $\delta$-close to the strategy.
Let $\ket{\phi_{\lambda}'}=\ket{\phi_{\lambda}}\otimes\ket{\Phi^+}$. The two states, $(U_{AQ}^{x}\otimes U_{B}^{y})\ket{\psi'}$ and $(U_{AQ}^{x'}\otimes U_B^{y'})\ket{\phi_{\lambda}'}$, are close,
\begin{equation}
\begin{split}
    &\Delta((U_{AQ}^{x}\otimes U_{B}^{y})\ket{\psi'},(U_{AQ}^{x'}\otimes U_B^{y'})\ket{\phi_{\lambda}'})\\
    \leq&\norm{(U_{AQ}^{x}\otimes U_{B}^{y})\ket{\psi'}-(U_{AQ}^{x'}\otimes U_B^{y'})\ket{\phi_{\lambda}'}}_2\\
    \leq&\norm{(U_{AQ}^{x}\otimes U_{B}^{y})\ket{\psi'}-(U_{AQ}^{x'}\otimes U_{B}^{y})\ket{\psi'}}_2+\norm{(U_{AQ}^{x'}\otimes U_{B}^{y})\ket{\psi'}-(U_A^{x'}\otimes U_B^{y'})\ket{\psi'}}_2\\
    &+\norm{(U_{AQ}^{x'}\otimes U_B^{y'})\ket{\psi'}-(U_{AQ}^{x'}\otimes U_B^{y'})\ket{\phi_{\lambda}'}}_2\\
    \leq&\norm{(U_{AQ}^{x}-U_{AQ}^{x'})\otimes\mathbb{I}}_{op}\norm{(\mathbb{I}\otimes U_{B}^{y})\ket{\psi'}}_2+\norm{\mathbb{I}\otimes(U_{B}^{y}-U_B^{y'})}_{op}\norm{(U_{AQ}^{x'}\otimes\mathbb{I} )\ket{\psi'}}_2+\norm{\ket{\psi'}-\ket{\phi_{\lambda}'}}_2\\
    \leq&3\delta.
\end{split}
\end{equation}
Therefore, $\ket{\psi^{x'y'}}\in\vars{S}^{\cth,\eta}_{f(x,y),3\delta}$, for which the assigned net returns value $g(x',y',\lambda)=[m]\setminus\{j\}$ such that $\forall \ket{\psi_j}\in\vars{S}_{j}$, $\Delta(\ket{\psi^{x'y'}},\ket{\psi_j})\geq\tilde{\delta}-3\delta$.
Since $\tilde{\delta}>6\delta$, it is clear that $\ket{\psi^{x'y'}}\notin\vars{S}^{\cth,\eta}_{j,3\delta}$ and $f(x,y)\neq j$.
This provides the correct assignment of $f(x,y)\in g(x',y',\lambda)$.
We note that the $\delta$-net for unitaries is defined by norm $\norm{\cdot}_{opvec}$, and $\norm{u-v}_{opvec}=\norm{U-V}_{op}$, where $u$ and $v$ are vector representations of the unitaries $U$ and $V$.
Since this works for all $l$ pairs of $(x,y)$, we can conclude that $g$ is a valid classical rounding.

The sizes of the discretized sets are given by $\log_2\abs{\vars{N}_S}\leq2^{2q+2}\log_2(1+\frac{2}{\delta})$ for the state, $\log_2\abs{\vars{N}_A}\leq2^{2q+3}\log_2(1+\frac{2}{\delta})$ for Alice's unitary acting on $q+1$ qubits, and $\log_2\abs{\vars{N}_B}\leq2^{2q+1}\log_2(1+\frac{2}{\delta})$  for Bob's unitary acting on $q$ qubits.
Therefore, we select the maximum set size as $k$, with
\begin{equation}
    k=2^{2q+3}\log_2\Bigg\lceil 2+\frac{12}{\tilde{\delta}}\Bigg\rceil,
\end{equation}
where we note that $\delta$ can be selected close enough to $\frac{\tilde{\delta}}{6}$ to cause a single bit change for the term within the logarithm, and the ceiling function ensures that $2^k$ (input size to classical rounding function $g$) is an integer.

\subsection{Proof of Theorem~\ref{thm:nu_value_compute}}

For a fixed classical rounding of size $k$, it is possible to implement a maximum of $2^k\times (2^{k})^{2^n}\times (2^{k})^{2^n}=2^{k(2^{n+1}+1)}$ functions.
Therefore, for a random function $f$, the probability that we can find a suitable $f_A$ and $f_B$ such that for at most $2^{2n}\nu$ $(x,y)$ pairs (which can be grouped into a set $\vars{S}$ with size $\abs{\vars{S}}\leq 2^{2n}\nu$), $f(x,y)\notin g(f_A(x),f_B(y),\lambda)$ is
\begin{equation}
\begin{split}
    &\Pr[f:\exists f_A,f_B,\lambda,\exists\vars{S},\abs{\vars{S}}\leq 2^{2n}\nu\quad s.t.\quad\substack{\forall (x,y)\in\vars{S}:f(x,y)\notin g(f_A(x),f_B(y),\lambda)\\\forall(x,y)\notin\vars{S}:f(x,y)\in g(f_A(x),f_B(y),\lambda)}]\\
    =&\frac{\abs{\{f:\exists f_A,f_B,\lambda,\exists\vars{S},\abs{\vars{S}}\leq 2^{2n}\nu\quad s.t.\quad\substack{\forall (x,y)\in\vars{S}:f(x,y)\notin g(f_A(x),f_B(y),\lambda)\\\forall(x,y)\notin\vars{S}:f(x,y)\in g(f_A(x),f_B(y),\lambda)}\}}}{\abs{\{f:\{0,1\}^{2n}\rightarrow[m]\}}}\\
    \leq&\frac{1}{m^{2^{2n}}}\sum_{i=1}^{2^{2n}\nu}\sum_{\vars{S}:\abs{\vars{S}}=i}\sum_{f_A,f_B,\lambda}\abs{\{f:\substack{\forall (x,y)\in\vars{S}:f(x,y)\notin g(f_A(x),f_B(y),\lambda)\\\forall(x,y)\notin\vars{S}:f(x,y)\in g(f_A(x),f_B(y),\lambda)}\}}\\
    \leq&\frac{1}{m^{2^{2n}}}\sum_{i=1}^{2^{2n}\nu}\sum_{\vars{S}:\abs{\vars{S}}=i}\sum_{f_A,f_B,\lambda}(m-1)^{2^{2n}-i}\\
    \leq&\left(\frac{m-1}{m}\right)^{2^{2n}}2^{k(2^{n+1}+1)}\sum_{i=1}^{2^{2n}\nu}\begin{pmatrix} 2^{2n} \\ i \end{pmatrix}\\
    \leq&2^{2^{2n}[\log_2(1-\frac{1}{m})+h_b(\nu)]+2^{2q+3}(2^{n+1}+1)\log_2\big\lceil2+\frac{12}{\tilde{\delta}}\big\rceil}\\
    \leq&2^{2^{2n}[\log_2(1-\frac{1}{m})+h_b(\nu)+2^{5-2q_0}\log_2\big\lceil2+\frac{12}{\tilde{\delta}}\big\rceil]},
\end{split}
\end{equation}
where the third line upper bounds the counting of functions where suitable $f_A$, $f_B$, $\lambda$ and $\vars{S}$ can be found by the sum of functions that can satisfy the required conditions for every valid choice of $f_A$, $f_B$, $\lambda$ and $\vars{S}$.
The fourth line notes that for $(x,y)\in\vars{S}$, there is only one valid value that $f(x,y)$ can take while for $(x,y)\notin\vars{S}$, there are $(m-1)$ values.
The fifth line upper bounds $(m-1)^{-i}\leq 1$, while the sixth line upper bounds the sum of binomial coefficients~\cite{vanLint1999} using binary entropy $h_b(x)$ and substitutes $k$.
We want to select a suitable $\nu$ such that being able to find $f_A$ and $f_B$ is the exception that occurs with probability not more than $2^{-\beta}$, which implies
\begin{equation}
    \nu= h_b^{-1}\left\{\log_2\left(\frac{m}{m-1}\right)-2^{5-2q_0}\log_2\Bigg\lceil2+\frac{12}{\tilde{\delta}}\Bigg\rceil-\frac{\beta}{2^{2n}}\right\}
\end{equation}
where we note $\nu\in[0,\frac{1}{2}]$ for the inverse of binary entropy to exist.

\subsection{Improved trace distance bounds}
Based on the security analysis against unentangled adversaries, we know that the adversaries cannot obtain a high score if they share a fixed state. If a given state gives a high score for one particular basis, the state will give a low score for another basis, which limits the maximum score that the adversaries can get using that state. Therefore, to obtain a high average score, the state preparation needs to be different for different inputs. Here, we want to formalize this intuition. Specifically, we aim to show that if the adversaries obtain a high score for all possible basis choices, they need to prepare states that are far apart (in trace distance), for at least one pair of basis choices.

We begin by proving the following approximation
\begin{lemma}[Approximation to the function $\sqrt{1-x^2}$]\label{lemma: linear approximation}
    Let $t \in \mathbb{N}$, $x \in [0,1]$, and $\{x_0, ..., x_t\}$ be arbitrary nodes such that $x_0 = 0$, $x_t = 1$, and $x_{k-1} < x_k < x_{k+1}$ for all $0 < k < t$. Then, we have
    \begin{equation}
        \sqrt{1 - x^2} \geq \min_{k \in [t]} \, m_k x + c_k,
    \end{equation}
    where, for all $k \in [t]$,
    \begin{equation} \label{eq: gradient and intercept}
        \begin{split}
            m_k &= \frac{\sqrt{1 - x_k^2} - \sqrt{1 - x_{k-1}^2}}{x_k - x_{k-1}},\\
            c_k &= \sqrt{1 - x_k^2} - m_k x_k.
        \end{split}
    \end{equation}
\end{lemma}
\begin{proof}
    To prove the claim, we first show that the function $\sqrt{1-x^2}$ is concave within the interval $[0,1]$. We do this by showing that the second derivative is negative,
    \begin{equation*}
        \frac{\mathrm{d}^2}{\mathrm{d}x^2} \sqrt{1-x^2} = -\frac{1}{(1-x^2)^{3/2}} < 0, \qquad \forall x \in [0,1].
    \end{equation*}
    Since its second derivative is negative in the interval $[0,1]$, the function $\sqrt{1-x^2}$ is concave in that interval.

    Next, since the curve $y = \sqrt{1 - x^2}$ is concave, we have that for any $x_{k-1}, x_k \in [0,1]$ such that $x_{k-1} < x_k$ and for any $x \in [x_{k-1}, x_k] $ and $\lambda = (x_k - x)/(x_k - x_{k-1})$:
    \begin{align*}
         \sqrt{1-x^2} = \sqrt{1-\left( \lambda x_{k-1} + (1-\lambda) x_k\right)^2} &\geq \lambda \sqrt{1-x_{k-1}^2} + (1-\lambda)\sqrt{1-x_k^2}\\
        &= \lambda \left(\sqrt{1-x_{k-1}^2} - \sqrt{1 - x_k^2} \right) + \sqrt{1 - x_k^2}\\
        &= \frac{\sqrt{1-x_{k-1}^2} - \sqrt{1 - x_k^2} }{x_k - x_{k-1}} (x_k - x) + \sqrt{1 - x_k^2}\\
        &= \frac{\sqrt{1-x_k^2} - \sqrt{1 - x_{k-1}^2} }{x_k - x_{k-1}} x + \left(\sqrt{1 - x_k^2} - \frac{\sqrt{1-x_k^2} - \sqrt{1 - x_{k-1}^2} }{x_k - x_{k-1}} x_k \right)\\
        &=: m_k x + c_k
    \end{align*}
    
    Also, by construction, the line $y = m_k x + c_k$ intersects the curve $y = \sqrt{1-x^2}$ at the points $\left(x_{k-1}, \sqrt{1-x_{k-1}^2}\right)$ and $\left(x_{k}, \sqrt{1-x_{k}^2}\right)$. Thus, for any $x \notin [x_{k-1}, x_{k}]$, we have $\sqrt{1-x^2} < m_{k} x + c_{k}$. Therefore, for any $x \in [0,1]$, $\min_{k \in [t]} \, m_k x + c_k$ can be attained by $k^*$ such that $x \in [x_{k^*-1}, x_{k^*}]$, in which case the bound $\sqrt{1-x^2} \geq m_{k^*} x + c_{k^*}$ applies. This proves the claim of the lemma.
\end{proof}

Based on the above approximation, we can start deriving some lower bound on the trace distance. First, we consider the lower bound on the average trace distance between different state preparations, given that each state gives a high score for each basis choice,
\begin{theorem}[Average trace distance bound] \label{thm: average trace distance bound}
    Let $t \in \mathbb{N}$, $\{x_0, ..., x_t\}$ be arbitrary nodes in the interval $[0,1]$ with $x_0 = 0$ and $x_t = 1$, and $\{m_k\}_{k \in [t]}$ and $\{c_k\}_{k \in [t]}$ be defined as in Eq.~\eqref{eq: gradient and intercept}. 
    
    Let $\eta \in [0,1]$ and $\cth \in [-1, 1]$. For each $i \in [m]$, let $\vars{S}_i^{\cth, \eta}$ be the set defined in Definition~\ref{def: C-eta set} for basis $i$:
    \begin{equation*}
        \vars{S}_{i}^{\cth, \eta} = \left\{\ket{\Psi_i} = (\mathbb{I}_V \otimes U^{i}_{Q \rightarrow AB}) \ket{\Phi^+}_{VQ}: \exists\{A^{i}_z\}_z, \{B^{i}_z\}_z, \, \mathrm{s.t.} \, C_i \geq \cth, p_t = \eta, p_a = 0 \right\},
    \end{equation*}
    where the set is characterized by the parameters $\cth$ and $\eta$. Denote $\ket{\psi_{i0}} = U^{i}_{Q\rightarrow AB} \ket{0}$ and $\ket{\psi_{i1}} = U^i_{Q \rightarrow AB} \ket{1}$. If $\ket{\Psi_i} \in \vars{S}^{\cth,\eta}_i$ for all $i \in [m]$, for any $j \in [m]$, then $\frac{1}{m-1}\sum_{i \neq j} \Delta(\ket{\Psi_i}, \ket{\Psi_j})$ is lower bounded by
    \begin{equation} \label{eq: sdp six-state}
    \begin{split}
        \tilde{\delta} = \min_{\vec{k} \in [t]^{m-1}} \min_{G^\ell} \quad & \frac{1}{m-1} \sum_{i \neq j} \frac{1}{2} m_{k_i} \mathrm{Re}\left\{ \braket{\psi_{i0}}{\psi_{j0}} + \braket{\psi_{i1}}{\psi_{j1}} \right\} + c_{k_i}\\
        \mathrm{s.t.} \quad& G^\ell\geq0\\
        &\frac{1}{2} \mathrm{Re}\left\{\braket{\psi_{i0}}{\psi_{j0}} + \braket{\psi_i1}{\psi_{j1}} \right\} \geq 0, \quad \forall i \neq j \\
    &\braket{\psi_{is}}{\psi_{is'}}=\delta_{ss'},\,\forall i\in [m]\\
    &\frac{1}{2}\sum_{s=0}^1 \mel{\psi_{is}}{A_{\perp}^i}{\psi_{is}} = 1 - \eta, \, \forall i \in [m]\\
    &\frac{1}{3}\Bigg[\frac{\cos \theta_i}{2} \sum_{s=0}^1 \sum_{a=0}^1 (-1)^{a+s} \mel{\psi_{is}}{A_a^i}{\psi_{is}} + \sin \theta_i \cos \phi_i \mathrm{Re} \left\{ \mel{\psi_{i0}}{(A_0^i - A_1^i)}{\psi_{i1}} \right\} \\
    &\qquad- \sin \theta_i \sin \phi_i \mathrm{Im} \left\{ \mel{\psi_{i0}}{(A_0^i - A_1^i)}{\psi_{i1}}\right\} \Bigg] \geq \cth, \, \forall i \in [m]\\
    &\mel{\psi_{is}}{A^i_a B^i_{b}}{\psi_{is}}=0,\,\forall a\neq b, \forall s\in\{0,1\}, \forall i \in [m]\\
    &[A^i_a,B^{i'}_{b}]=0,\, \forall i,i'\in [m], \, \forall a,b\in\{0,1,\perp\}
    \end{split}
    \end{equation}
\end{theorem}
\begin{proof}
    First, we note that for two pure states, $\ket{\Psi_i}$ and $\ket{\Psi_j}$, the trace distance between them is related to their inner product by~\cite{Wilde2017}
    \begin{equation*}
        \Delta(\ket{\Psi_i}, \ket{\Psi_j}) = \sqrt{1 - \abs{\braket{\Psi_i}{\Psi_j}}^2}.
    \end{equation*}
    Next, for each state $\ket{\Psi_i}$ with $i \neq j$, we can apply an arbitrary global phase without affecting the trace distance and the statistics. Consequently, without any loss of generality, we can take the inner products to be real and positive
    \begin{equation*}
        \abs{\braket{\Psi_i}{\Psi_j}} = \braket{\Psi_i}{\Psi_j} \geq 0, \quad \forall i \neq j.
    \end{equation*}
    For each $i \neq j$, we can lower bound the trace distance with the function
    \begin{align*}
        \Delta(\ket{\Psi_i}, \ket{\Psi_j}) &= \sqrt{1 - \braket{\Psi_i}{\Psi_j}^2}\\
        &\geq \min_{k_i \in [t]} m_{k_i}  \braket{\Psi_i}{\Psi_j} + c_{k_i}, \quad \forall i \neq j.
    \end{align*}
    Here, the equality comes from the fact that we can set the inner product $\braket{\Psi_i}{\Psi_j}$ to be real (and positive) without any loss of generality. The inequality comes from the application of Lemma~\ref{lemma: linear approximation}.
    
    We note that the above lower bound on the trace distance holds for any arbitrary fixed pair of states $\ket{\Psi_i}$ and $\ket{\Psi_j}$ but the right hand side of the bound clearly depends on the choice of states (via their inner product). To obtain a state-independent lower bound, we have to perform a minimisation over any states $\{\ket{\Psi_i} \in \vars{S}_i^{\cth, \eta}\}_{i \in [m]}$. In other words, we want to solve for 
    \begin{equation*}
        \min_{\{\ket{\Psi_i} \in \vars{S}_i^{\cth, \eta}\}_{i \in [m]}} \min_{\{k_i \in [t]\}_{i \neq j}} \frac{1}{m-1} \sum_{i \neq j} \frac{1}{2} m_{k_i} \mathrm{Re}\left\{ \braket{\psi_{i0}}{\psi_{j0}} + \braket{\psi_{i1}}{\psi_{j1}} \right\} + c_{k_i}.
    \end{equation*}
    
    Now, we note that since both optimizations are minimization, the order of the minimization can be swapped: we can minimize over the states first, then minimize over the choice of $\vec{k} = \{k_i: i \neq j\}$. Moreover, for each fixed choice of $\vec{k}$, the minimization over the states can be bounded using the SDP relaxation introduced in Ref.~\cite{wang2019characterising}, since the objective function and the inequalities that characterize the sets $\{\vars{S}_i^{\cth, \eta}\}_{i\in [m]}$ are all linear in $G^\ell$, where $G^\ell$ is the Gram matrix formed by the (sub-normalized) vectors $\{\ket{\xi_{is}^r} = \gamma_r \ket{\psi_{is}} \}_{i, s, r}$ where $\gamma_r \in \Gamma^\ell$. 
    
    Therefore, the optimization over the states can be replaced by the optimization over the Gram matrix $G^\ell$ satisfying the linear constraints that define the sets $\vars{S}_i^{\cth, \eta}$. For each choice of $\vec{k}$ this will give a valid state-independent lower bound on $\frac{1}{m-1}\sum_{i \neq j} \frac{1}{2} m_{k_i} \mathrm{Re}\left\{ \braket{\psi_{i0}}{\psi_{j0}} + \braket{\psi_{i1}}{\psi_{j1}} \right\} + c_{k_i}$. The state-independent lower bound on the average trace distance is finally obtained by choosing the $\vec{k}$ that minimizes the expression.
\end{proof}

Theorem~\ref{thm: average trace distance bound} gives a lower bound on the average trace distance $\frac{1}{m-1} \sum_{i \neq j} \Delta(\ket{\Psi_i}, \ket{\Psi_j})$ for some choice of $j \in [m]$ which holds whenever $\ket{\Psi_i} \in \vars{S}_i^{\cth, \eta}$ for all $i \in [m]$. However, this also gives Corollary~\ref{cor: existence trace distance bound} (restated below for convenience) which is easier to use in the security proof.
\begin{customthm}{\ref{cor: existence trace distance bound}}
Let $j \in [m]$. For any state $\ket{\Psi_j} \in \vars{S}_j^{\cth, \eta}$, there exists $i \in [m]$ such that $i \neq j$ and for any state $\ket{\Psi_i} \in \vars{S}_i^{\cth, \eta}$, we have $\Delta(\ket{\Psi_i}, \ket{\Psi_j})\geq \tilde{\delta}$, where $\tilde{\delta}$ is the solution to the optimization problem given by Eq.~\eqref{eq: sdp six-state}.
\end{customthm}
\begin{proof}
    From Theorem~\ref{thm: average trace distance bound}, if $\ket{\Psi_i} \in \vars{S}_i^{\cth, \eta}$ for all $i \in [m]$, we know that the average trace distance is lower bounded
    \begin{equation*}
        \frac{1}{m-1} \sum_{i \neq j} \Delta(\ket{\Psi_i}, \ket{\Psi_j}) \geq \tilde{\delta}.
    \end{equation*}
    Since the average trace distance is lower bounded by $\tilde{\delta}$, we know that at least one of the terms must be lower bounded by $\tilde{\delta}$, i.e., there exists $i^* \neq j$ such that $\Delta(\ket{\Psi_{i^*}}, \ket{\Psi_j}) \geq \tilde{\delta}$.
    
    We also note that the lower bound in Theorem~\ref{thm: average trace distance bound} is state-independent, and once the state $\ket{\Psi_j}$ is fixed, for each $i \neq j$, the optimization over $\ket{\Psi_i}$ is independent of the other $i' \neq j$. Therefore, for a fixed state $\ket{\Psi_j}$, it is optimal to choose $\{\ket{\Psi_i}\}_{i \neq j}$ that are individually the closest in trace distance to $\ket{\Psi_j}$. Since for all $i \neq j$ (including $i^*$), the optimizer is the closest state in $\vars{S}_{i}^{\cth, \eta}$ to the state $\ket{\Psi_j}$, the conclusion that there exists an index $i^* \neq j$ such that $\Delta(\ket{\Psi_{i^*}}, \ket{\Psi_j}) \geq \tilde{\delta}$ must also apply to any other states in $\vars{S}_{i^*}^{\cth,\eta}$ that may not be the closest to $\ket{\Psi_j}$. This implies the claim.
\end{proof}

Theorem~\ref{thm: average trace distance bound} and Corollary~\ref{cor: existence trace distance bound} are derived when considering the case where $\ket{\Psi_i} \in \vars{S}_i^{\cth, \eta}$ for all $i \in [m]$. This case is useful when considering $q$-qubit restricted strategies. It turns out that when considering $q$-qubit general strategies, it is more convenient to characterize the performance of the strategy by the normalized score $\hat{C}_i = C_i / p_t$. Fortunately, the normalized score is also compatible with the SDP relaxation~\cite{wang2019characterising}, and thus we have the following analogues of Theorem~\ref{thm: average trace distance bound} and Corollary~\ref{cor: existence trace distance bound} for the case where the sets are characterized by the normalized score,
\begin{theorem} \label{thm: average trace distance bound for normalised score}
    Let $t \in \mathbb{N}$, $\{x_0, ..., x_t\}$ be arbitrary nodes in the interval $[0,1]$ with $x_0 = 0$ and $x_t = 1$, and $\{m_k\}_{k \in [t]}$ and $\{c_k\}_{k \in [t]}$ be defined as in Eq.~\eqref{eq: gradient and intercept}. 
    
    Let $\etath \in [0,1]$ and $\cthnorm \in [-1, 1]$. For each $i \in [m]$, let $\hat{\vars{S}}_i^{\cthnorm, \etath}$ be the set defined in Definition~\ref{def: normalised C-eta set} for basis $i$:
    \begin{equation*}
        \hat{\vars{S}}_{i}^{\cthnorm, \etath} = \left\{\ket{\Psi_i} = (\mathbb{I}_V \otimes U^{i}_{Q \rightarrow AB}) \ket{\Phi^+}_{VQ}: \exists\{A^{i}_z\}_z, \{B^{i}_z\}_z, \, \mathrm{s.t.} \, C_i \geq \cthnorm p_t, p_t \geq \etath, p_a = 0 \right\},
    \end{equation*}
    where the set is characterized by the parameters $\cthnorm$ and $\etath$. Denote $\ket{\psi_{i0}} = U^{i}_{Q\rightarrow AB} \ket{0}$ and $\ket{\psi_{i1}} = U^i_{Q \rightarrow AB} \ket{1}$. If $\ket{\Psi_i} \in \hat{\vars{S}}^{\cthnorm,\etath}_i$ for all $i \in [m]$, for any $j \in [m]$, then $\frac{1}{m-1}\sum_{i \neq j} \Delta(\ket{\Psi_i}, \ket{\Psi_j})$ is lower bounded by
    \begin{equation} \label{eq: sdp six-state normalised}
    \begin{split}
        \tilde{\delta} = \min_{\vec{k} \in [t]^{m-1}} \min_{G^\ell} \quad & \frac{1}{m-1} \sum_{i \neq j} \frac{1}{2} m_{k_i} \mathrm{Re}\left\{ \braket{\psi_{i0}}{\psi_{j0}} + \braket{\psi_{i1}}{\psi_{j1}} \right\} + c_{k_i}\\
        \mathrm{s.t.} \quad& G^\ell\geq0\\
        &\frac{1}{2} \mathrm{Re}\left\{\braket{\psi_{i0}}{\psi_{j0}} + \braket{\psi_i1}{\psi_{j1}} \right\} \geq 0, \quad \forall i \neq j \\
    &\braket{\psi_{is}}{\psi_{is'}}=\delta_{ss'},\,\forall i\in [m]\\
    &\frac{1}{2}\sum_{s=0}^1 \mel{\psi_{is}}{A_{\perp}^i}{\psi_{is}} \leq 1 - \etath, \, \forall i \in [m]\\
    &\frac{1}{3}\Bigg[\frac{\cos \theta_i}{2} \sum_{s=0}^1 \sum_{a=0}^1 (-1)^{a+s} \mel{\psi_{is}}{A_a^i}{\psi_{is}} + \sin \theta_i \cos \phi_i \mathrm{Re} \left\{ \mel{\psi_{i0}}{(A_0^i - A_1^i)}{\psi_{i1}} \right\} \\
    &\qquad- \sin \theta_i \sin \phi_i \mathrm{Im} \left\{ \mel{\psi_{i0}}{(A_0^i - A_1^i)}{\psi_{i1}}\right\} \Bigg] \geq \frac{\cthnorm}{2} \sum_{s=0}^1 \sum_{a=0}^1 \mel{\psi_{is}}{A_a^i}{\psi_{is}}, \, \forall i \in [m]\\
    &\mel{\psi_{is}}{A^i_a B^i_{b}}{\psi_{is}}=0,\,\forall a\neq b, \forall s\in\{0,1\}, \forall i \in [m]\\
    &[A^i_a,B^{i'}_{b}]=0,\, \forall i,i'\in [m], \, \forall a,b\in\{0,1,\perp\}
    \end{split}
    \end{equation}
\end{theorem}
\begin{corollary}\label{cor: existence trace distance bound for normalised score}
    Let $j \in [m]$. For any state $\ket{\Psi_j} \in \hat{\vars{S}}_j^{\cthnorm, \etath}$, there exists $i \in [m]$ such that $i \neq j$ and for any state $\ket{\Psi_i} \in \hat{\vars{S}}_i^{\cthnorm, \etath}$, we have $\Delta(\ket{\Psi_i}, \ket{\Psi_j})\geq \tilde{\delta}$, where $\tilde{\delta}$ is given by Eq.~\eqref{eq: sdp six-state normalised}.
\end{corollary}

The proofs for Theorem~\ref{thm: average trace distance bound for normalised score} and Corollary~\ref{cor: existence trace distance bound for normalised score} are similar to the ones for Theorem~\ref{thm: average trace distance bound} and Corollary~\ref{cor: existence trace distance bound}, since those proofs only require that the constraints that characterize the sets $\vars{S}_i^{\cth,\eta}$ and $\hat{\vars{S}}_i^{\cthnorm,\etath}$ are both linear in $G^{\ell}$. Therefore, the only difference will be the explicit constraints for each set, whereas the remaining part of the proofs remain unchanged.

\subsection{Proof of Theorem~\ref{thm:purification}}
\label{app:Partial_Purification}

The maximum achievable score of an adversary can be expressed as an optimization problem:
\begin{equation}
\begin{split}
    \max_{\substack{\rho_{AB}^r\in\vars{S}_{2q}\\\channel_{AQ}^{xr}\in\vars{C}_{q+1},\channel_B^{yr}\in\vars{C}_q\\\{A_z^{xyr}\},\{B_z^{xyr}\}}}&\sum_{rszxy}\frac{(-1)^{s+z}p_r}{\abs{\vars{X}}\abs{\vars{Y}}}\Tr[(\Lambda_s^{f(x,y)}\otimes A_{z}^{xyr}\otimes B_{z}^{xyr})\rho_{A'B'V}^{xyr}]\\
    subj.\, to\quad&\sum_{xyr}\frac{p_r}{\abs{\vars{X}}\abs{\vars{Y}}}\Tr[(\mathbb{I}_V\otimes A_{\perp}^{xyr}\otimes B_{\perp}^{xyr})\rho_{A'B'V}^{xyr}]=1-\eta\\
    &\Tr[(\mathbb{I}_V\otimes A_{z}^{xyr}\otimes B_{\bar{z}}^{xyr})\rho_{A'B'V}^{xyr}]=0,\,\forall z,x,y,r
\end{split},
\end{equation}
where
\begin{equation}
    \rho_{A'B'V}^{xyr}=\channel_{AQ}^{xr}(\dyad{\Phi^+}_{VQ}\otimes\channel_B^{yr}(\rho_{AB}^r)).
\end{equation}

The purification can be performed in two steps.
We first purify the quantum system by introducing a quantum system $P$ of dimension $2^{2q}$, and selecting a purification such that $\Tr_P[\dyad{\psi_r}_{ABP}]=\rho^r_{AB}$~\cite{Wilde2017}.
Furthermore, we can lift the quantum channels to higher dimensional unitaries by introducing quantum systems $P_A$ and $P_B$, of dimension $2^{2(q+1)}$ and $2^{2q}$ respectively, and selecting the unitaries such that $\Tr_{P_A}[U^{xr}_{AQP_A}(\rho_{AQ}\otimes\dyad{0}_{P_A}) U_{AQP_A}^{xr\dagger}]=\channel_{AQ}^{xr}(\rho)$ for any state $\rho_{AQ}$ and $\Tr_{P_B}[U^{yr}_{BP_B}(\rho_B\otimes\dyad{0}_{P_B}) U_{BP_B}^{yr\dagger}]=\channel_B^{yr}(\rho_B)$ for any state $\rho_B$~\cite{Wilde2017}.
We note that the states and channels for different values of $r$ can be separately purified, with the purification systems being of the same Hilbert space, since they can be distinguished from the value of $r$.
Collating the changes, this leads to the partially purified state
\begin{equation*}
    \rho_{A'B'VPP_AP_B}=\sum_rp_r\dyad{r}_R\otimes[(U_{AQP_A}^{xr}\otimes U_{BP_B}^{yr})(\dyad{\Phi^+}_{VQ}\otimes\dyad{\psi^r}_{ABP}\otimes\dyad{0}_{P_AP_B})(U_{AQP_A}^{xr}\otimes U_{BP_B}^{yr})^{\dagger}]
\end{equation*}
and $\Tr_{PP_AP_B}[\rho_{RA'B'VPP_AP_B}]=\rho_{RA'B'V}$.

As a consequence, any mixed state strategy can be expressed as a pure state strategy of the higher dimension.
As such, optimizing over the larger set of higher dimension restricted state strategy can only lead to a larger maximum score, given by
\begin{equation}
\label{eq:Winning_prob_opt}
\begin{split}
    \max_{\substack{\ket{\psi^r}_{ABP}\in\vars{S}^p_{4q}\\U_{AQP_A}^{xr}\in\vars{C}^U_{3(q+1)},U_{BP_B}^{yr}\in\vars{C}^U_{3q}\\\{A_z^{xyr}\},\{B_z^{xyr}\}}}&\sum_{rszxy}\frac{(-1)^{s+z}p_r}{\abs{\vars{X}}\abs{\vars{Y}}}\Tr[(\Lambda_s^{f(x,y)}\otimes A_{z}^{xyr}\otimes B_{z}^{xyr})\dyad{\Psi_{xyr}}]\\
    subj.\, to\quad&\sum_{xyr}\frac{p_r}{\abs{\vars{X}}\abs{\vars{Y}}}\Tr[(\mathbb{I}_V\otimes A_{\perp}^{xyr}\otimes B_{\perp}^{xyr})\dyad{\Psi_{xyr}}]=1-\eta\\
    &\Tr[(\mathbb{I}_V\otimes A_{z}^{xyr}\otimes B_{\bar{z}}^{xyr})\dyad{\Psi_{xyr}}]=0,\,\forall z,x,y,r
\end{split},
\end{equation}
where
\begin{equation}
    \ket{\Psi_{xyr}}_{A'B'V}=(U_{AQP_A}^{xr}\otimes U_{BP_B}^{yr})(\ket{\Phi^+}_{VQ}\otimes\ket{\psi_r}_{ABP}).
\end{equation}

\subsection{Proof of Theorem~\ref{thm:ABE Bound}}

Consider the general strategy, with input and shared randomness dependent score and transmission
\begin{equation}
\begin{gathered}
    C_{rxy}=\sum_{s,z\in\{0,1\}}(-1)^{s+z}\mel{\Psi^{xy}}{\Lambda_s^{\alpha'}\otimes A_z^{xyr}\otimes B_z^{xyr}}{\Psi^{xy}}\\
    \eta_{rxy}=1-\mel{\Psi^{xy}}{\mathbb{I}\otimes A_{\perp}^{xyr}\otimes B_{\perp}^{xyr}}{\Psi^{xy}},
\end{gathered}
\end{equation}
where we set $\alpha'=f(x,y)$ and WLOG, we let the state be pure and measurement operators be projective.
We can simplify
\begin{equation}
    \sum_s(-1)^s\Lambda_s^{\alpha'}=\vec{a}\cdot\vec{\sigma},
\end{equation}
where $\vec{a}=[\cos\theta_{\alpha'},\sin\theta_{\alpha'}\cos\varphi_{\alpha'},\sin\theta_{\alpha'}\sin\varphi_{\alpha'}]$ and $\vec{\sigma}$ is a vector of Pauli matrices.
Let us define $\hat{AB}^{xyr}=\sum_{z\in\{0,1\}}(-1)^zA_z^{xyr}\otimes B_z^{xyr}$, $A^{xyr}_{\top}=A^{xyr}_0+A^{xyr}_1$ and $B^{xyr}_{\top}=B^{xyr}_0+B^{xyr}_1$.
Since $A^{xyr}_{\top}$ and $B^{xyr}_{\top}$ are projectors, we can define the sub-normalized state $\ket{\Psi_{\top}^{xy}}=(\mathbb{I}\otimes A^{xyr}_{\top}\otimes B^{xyr}_{\top})\ket{\Psi^{xy}}$, which is bounded by the transmission, $\braket{\Psi_{\top}^{xy}}{\Psi_{\top}^{xy}}\leq\eta_{rxy}$.
Noting that $\hat{AB}^{xyr}=(A^{xyr}_{\top}\otimes B^{xyr}_{\top})\hat{AB}^{xyr}(A^{xyr}_{\top}\otimes B^{xyr}_{\top})$, we can simplify the score
\begin{equation}
\begin{split}
    C_{rxy}=&\mel{\Psi^{xy}}{(\vec{a}\cdot\vec{\sigma})\otimes(A^{xyr}_{\top}\otimes B^{xyr}_{\top})\hat{AB}^{xyr}(A^{xyr}_{\top}\otimes B^{xyr}_{\top})}{\Psi^{xy}}\\
    =&\mel{\Psi^{xy}_{\top}}{(\vec{a}\cdot\vec{\sigma})\otimes\hat{AB}^{xyr}}{\Psi^{xy}_{\top}}\\
    \leq&\norm{\vec{a}\cdot\vec{\sigma}}\norm{\hat{AB}^{xyr}}\braket{\Psi^{xy}_{\top}}{\Psi^{xy}_{\top}}\\
    \leq&\eta,
\end{split}
\end{equation}
where $\norm{X}$ is the operator norm, and we note that $-\mathbb{I}\leq\vec{a}\cdot\vec{\sigma}\leq\mathbb{I}$ when $\abs{\vec{a}}=1$ (equivalent to a $\pm 1$ observable based on measurement in the $\vec{a}$ direction on the Bloch sphere), with the same for $\hat{AB}^{xyr}$ (contains projective measurements with eigenvalues $\pm 1$).

\subsection{Proof of Theorem~\ref{thm:sub_strategy_err}}

Define $l_{i,r}'=2^{-2n}l_{i,r}$.
The score of any sub-strategy with $(\cthnorm,\etath,l_{1,r},l_{2,r},l_{3,r},l_{4,r})$-partition can be computed by
\begin{equation}
    C_r=\sum_{i=1}^4l_{i,r}'C_{i,r},
\end{equation}
where $\eta_{i,r}$ and $C_{i,r}$ are the average transmission and conditional score in their respective partitions.
Similarly, we have that the average transmission of the sub-strategy is bounded
\begin{equation}
    \eta_r=\sum_{i=1}^4l_{i,r}'\eta_{i,r}
    \leq[(l_{1,r}'+l_{3,r}')+(l_{2,r}'+l_{4,r}')\etath]
    =(l_{1,r}'+l_{3,r}')(1-\etath)+\etath,
\end{equation}
noting that $\sum_{i=1}^4l_{i,r}=2^{2n}$ since the partitions sum to include all $(x,y)$ pairs.
We can thus simplify the score
\begin{equation}
\begin{split}
    C_r<& l_{1,r}'\eta_{1,r}+l_{2,r}'\eta_{2,r}+(l_{3,r}'\eta_{3,r} +l_{4,r}'\eta_{4,r})\cthnorm\\
    =&\eta_r-(l_{3,r}'+l_{4,r}')(1-\cthnorm)\\
    \leq&\eta_r-l_{3,r}'\etath(1-\cthnorm),
\end{split}
\end{equation}
where the second line substitutes the average transmission equation in and the final inequality uses the fact that $l_{4,r}'\geq 0$ and $\eta_{3,r}\geq\etath$.
We consider two cases, one where $1-\nu>\frac{\eta_r-\etath}{1-\etath}$ and one where $1-\nu\leq\frac{\eta_r-\etath}{1-\etath}$.
In the first case, we can simply bound $C_r\leq\eta_r$, which is true by definition of the score.
In the second case, the transmission bound and the assumption of $l_{1,r}'\leq 1-\nu$ allow us to lower bound the size of the third partition,
\begin{equation}
    l_{3,r}'\geq\frac{\eta_r-\etath}{1-\etath}-(1-\nu).
\end{equation}
Therefore, the score can be bounded by
\begin{equation}
    C_r\leq\eta_r-\etath(1-\cthnorm)\left[\frac{\eta_r-\etath}{1-\etath}-1+\nu\right].
\end{equation}
We note that when $1-\nu>\frac{\eta_r-\etath}{1-\etath}$, $\eta_r-\etath(1-\cthnorm)\left[\frac{\eta_r-\etath}{1-\etath}-1+\nu\right]\geq\eta_r$.
This allows us to combine the score across both cases as
\begin{equation}
    C_r\leq\min\left\{\eta_r-\etath(1-\cthnorm)\left[\frac{\eta_r-\etath}{1-\etath}-1+\nu\right],\eta_r\right\}
\end{equation}

\section{Acknowledegments}
This research received no external funding. Disclaimer: This paper was prepared for informational purposes by the Global Technology Applied Research center of JPMorgan Chase \& Co. This paper is not a product of the Research Department of JPMorgan Chase \& Co. or its affiliates. Neither JPMorgan Chase \& Co. nor any of its affiliates makes any explicit or implied representation or warranty and none of them accept any liability in connection with this paper, including, without limitation, with respect to the completeness, accuracy, or reliability of the information contained herein and the potential legal, compliance, tax, or accounting effects thereof. This document is not intended as investment research or investment advice, or as a recommendation, offer, or solicitation for the purchase or sale of any security, financial instrument, financial product or service, or to be used in any way for evaluating the merits of participating in any transaction.
All contributions by IWP and CL were made while working at JPMorgan Chase.

\bibliography{QPV_Ref}

\newpage

\foreach \x in {1,...,10}
{%
\clearpage
\includepdf[pages={\x}]{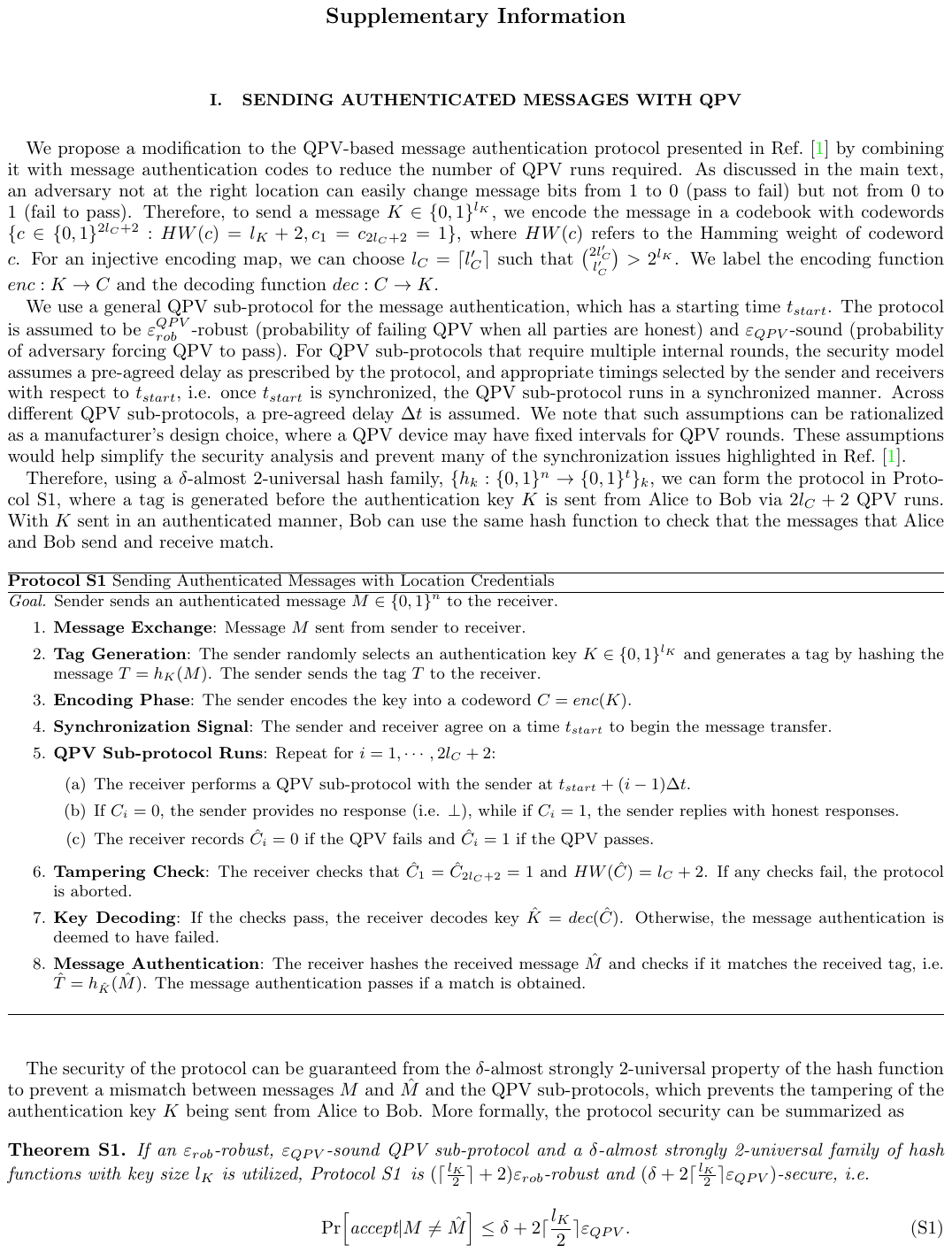} 
}


\end{document}